\documentclass[
  preprint,
  aps,
  prresearch,
  superscriptaddress,
  floatfix
]{revtex4-2}

\usepackage{amsmath,amssymb,bm}
\usepackage{amsthm} 
\usepackage[colorlinks=true,allcolors=black]{hyperref}
\hypersetup{pdfencoding=auto}
\usepackage{placeins}
\usepackage{braket}

\usepackage{makecell}

\usepackage{graphicx}
\usepackage{caption}
\usepackage{subcaption} 
\usepackage{overpic}    

\usepackage{rotating}

\theoremstyle{plain}
\newtheorem{theorem}{Theorem}[section]      
\newtheorem{proposition}[theorem]{Proposition} 

\theoremstyle{definition}
\newtheorem{requirement}{Requirement}          

\theoremstyle{remark}

\begin{document}

\title{Quantum-Circuit Framework for Two-Stage Stochastic Programming via QAOA Integrated with a Quantum Generative Neural Network}





\author{Taihei Kuroiwa}
 \email{k2443003@gl.cc.uec.ac.jp}
 \affiliation{Engineering Department of the University of
Electro-Communications, 182-8585, Chofu, Tokyo, Japan}
 \affiliation{Grid Inc., 107-0061, Tokyo, Japan}

\author{Daiki Yamazaki}
 \affiliation{Engineering Department of the University of
Electro-Communications, 182-8585, Chofu, Tokyo, Japan}

\author{Keita Takahashi}
 \affiliation{Grid Inc., 107-0061, Tokyo, Japan}

\author{Kodai Shiba}
 \affiliation{Grid Inc., 107-0061, Tokyo, Japan}

\author{Chih-Chieh Chen}
 \affiliation{Grid Inc., 107-0061, Tokyo, Japan}

\author{Tomah Sogabe}
\email[]{sogabe@uec.ac.jp}
\affiliation{Engineering Department of the University of
Electro-Communications, 182-8585, Chofu, Tokyo, Japan}
\affiliation{Grid Inc., 107-0061, Tokyo, Japan}
\affiliation{i-PERC, The University of Electro-Communications, 182-8585, Chofu, Tokyo, Japan}

\date{\today}

\begin{abstract}

Two-stage stochastic programming often discretizes uncertainty into scenarios, but scenario enumeration makes expected recourse evaluation scale at least linearly in the scenario count. We propose \emph{qGAN-QAOA}, a unified quantum-circuit workflow in which a pre-trained quantum generative adversarial network encodes the scenario distribution and QAOA optimizes first-stage decisions by minimizing the full two-stage objective, including expected recourse cost. With the qGAN parameters fixed after training, we evaluate the objective as the expectation value of a problem Hamiltonian and optimize only the QAOA variational parameters. We interpret non-anticipativity as a condition on measurement outcome statistics and prove that the first-stage measurement marginal is independent of the scenario. For uniformly discretized uncertainty, the diagonal operator encoding the uncertainty admits a sparse Pauli-$Z$ expansion via the Walsh--Hadamard transform, yielding polylogarithmic scaling of gate count and circuit depth with the number of scenarios.
Numerical experiments on the stochastic unit commitment problem (UCP) with photovoltaic (PV) uncertainty compare the expected cost of the proposed method with classical expected-value and two-stage stochastic programming baselines, demonstrating the effectiveness of qGAN-QAOA as a two-stage decision model.

\end{abstract}

\maketitle

\section{Introduction}
\label{sec:introduction}

In planning problems under uncertainty, forecast errors can necessitate additional adjustments and measures to avoid constraint violations, which may increase the operating cost of real systems. It is therefore important to model uncertainty as a probability distribution and to optimize decisions while accounting for the trade-off between risk and cost. A representative framework for addressing such problems is stochastic programming (SP) \cite{Dantzig_sp,shapiro2009lectures}, which has been widely applied to supply-chain design \cite{santoso2005_scnd}, financial asset allocation \cite{edirisinghe2007_mspo,consiglio2010_pfo}, and energy operations (e.g., stochastic unit commitment problem and virtual power plants) \cite{anjos2017_uc,emarati2020_vpp}. In particular, two-stage stochastic programming is a standard formulation in which a first-stage decision is made before observing uncertainty, and a second-stage recourse (adjustment) is optimized according to the realized uncertainty \cite{Dantzig_sp,shapiro2009lectures}.

In practice, uncertainty often arises in the form of continuous or multimodal distributions. However, in real-world applications the underlying population distribution is rarely known, and only observational data are typically available. In such cases, to treat uncertainty explicitly as a probability distribution, a common approach is to learn the distribution with generative models such as GANs and to generate samples of the uncertain quantity from a trained generator \cite{goodfellow2014gan}. In two-stage stochastic programming, it is standard to construct a scenario set based on the generated samples and to solve a discrete scenario-based problem via the sample average approximation (SAA), which approximates the expected recourse cost by a sample average \cite{shapiro2009lectures,kleywegt2002saa}.

In the SAA approach, it is known that the number of scenarios on the order of $N=O(\varepsilon^{-2})$ is required to reduce the approximation error of the expectation to $O(\varepsilon)$ \cite{shapiro2009lectures,kleywegt2002saa}. Furthermore, in the L-shaped method based on Benders decomposition, a representative algorithm for two-stage stochastic programming, the second-stage problem is solved for all scenarios at each iteration and the results are aggregated to update the master problem. Let $T_{\mathrm{2nd}}$ denote the computational cost of solving the second-stage problem for a single scenario; then the per-iteration computational burden increases at least as $O(NT_{\mathrm{2nd}})$ \cite{BENDERS,L-Shaped}. As a result, increasing the number of scenarios to improve accuracy becomes a practical bottleneck in both computation time and memory. This motivates the need for computational methods that can accurately treat uncertainty as a distribution while alleviating explicit dependence on the scenario count $N$.

In recent years, efforts to apply quantum computing to two-stage stochastic programming have been gaining momentum \cite{rotello2024_expected_value, mussig2025, xu2025_hybrid_qc_sp}. Rotello \textit{et al.} \cite{rotello2024_expected_value} proposed a framework in which candidate first-stage solutions are provided classically, while the expected recourse cost is evaluated by a quantum circuit by treating the second-stage problem in scenario-parallel fashion on a scenario register. In addition, hybrid frameworks have been reported in which quantum annealing is used for the first-stage combinatorial optimization and the second stage is handled by classical optimization \cite{xu2025_hybrid_qc_sp}. While these approaches incorporate quantum computation mainly for expectation evaluation or partial optimization, they do not treat the entire two-stage model in a unified manner as the optimization of a single variational quantum circuit.

In this work, we propose \emph{qGAN-QAOA}, a unified quantum-circuit workflow in which a pre-trained quantum generative adversarial network (qGAN) encodes the scenario distribution \cite{Lloyd2018QGAN,Dallaire2018QuGAN,Zoufal_2019_qGAN,Certo_2023_CqGAN,Tang2022_multi_qAGN}, and the quantum approximate optimization algorithm (QAOA) optimizes first-stage decisions by minimizing the two-stage objective, including the expected recourse cost \cite{farhi2014_QAOA,mcgeoch2014adiabatic,zhou2020quantum}. With the qGAN parameters fixed after training, the objective is evaluated as the expectation value of a problem Hamiltonian, and only the QAOA variational parameters are optimized. In two-stage stochastic programming, the non-anticipativity constraints require that first-stage decision variables be scenario-independent (see, e.g., \cite{Dantzig_sp,shapiro2009lectures}). However, to the best of our knowledge, existing quantum approaches have not explicitly discussed how to enforce or interpret these non-anticipativity constraints. Accordingly, we formulate the non-anticipativity constraints in terms of measurement outcome statistics and show that they hold because the marginal distribution of measurement outcomes on the first-stage register is independent of the scenario. Furthermore, by uniformly discretizing a continuous uncertainty and exploiting the Walsh--Hadamard transform (WHT), we propose a method to represent the random-variable operator $\hat{\xi}$---defined as a diagonal operator on the uncertainty register---via a sparse Pauli-$Z$ expansion.

Here, representations and circuit synthesis of diagonal operators based on the WHT/Walsh series have been systematized in prior work \cite{Welch2014DiagonalWalsh}, and efficient computation of Pauli-expansion coefficients via the WHT has also been reported \cite{Georges2025PauliFWHT}.
In this paper, leveraging sparse WHT coefficients of $\hat{\xi}$ realized by uniform discretization, we show that, for the number of scenarios $N$, the number of Pauli-$Z$ terms arising from $\hat{\xi}$ can be bounded by $O(\mathrm{poly}(\log N))$, thereby theoretically clarifying, from the viewpoint of dependence on $N$, that the gate count and circuit depth can scale polylogarithmically.
(although the overall circuit size may also depend on the number of terms associated with decision variables and compilation-related factors).

Our main contributions of this paper are summarized as follows.\\
(i) We propose qGAN-QAOA, a unified quantum-circuit workflow that combines scenario-distribution encoding by a pre-trained qGAN with variational optimization by QAOA, and evaluates the two-stage objective including expected recourse cost as a problem-Hamiltonian expectation value.\\
(ii) We formulate the non-anticipativity constraints in terms of measurement outcome statistics and show that they hold because the marginal distribution of measurement outcomes on the first-stage register is independent of the scenario in the proposed method.\\
(iii) Under uniform discretization of a continuous uncertainty, we show that the random-variable operator $\hat{\xi}$ (a diagonal operator) admits a sparse Pauli-$Z$ expansion based on the Walsh--Hadamard transform (WHT), and that the number of Pauli-$Z$ terms arising from $\hat{\xi}$ can be bounded by $O(\mathrm{poly}(\log N))$, thereby reducing the dependence on the scenario count $N$ induced by $\hat{\xi}$.\\
(iv) We organize the computational complexity of the proposed two-stage method in terms of a target accuracy $O(\varepsilon)$ and, by comparison with classical two-stage stochastic programming based on SAA, provide conditions under which the proposed method can be advantageous from the viewpoint of dominant terms as $\varepsilon\to 0$.\\
(v) As a case study of the stochastic unit commitment problem (UCP) with uncertainty in photovoltaic (PV) output, we evaluate the expected cost of the proposed method and compare it with the classical expected-value method and the two-stage stochastic programming approach, thereby demonstrating the effectiveness of qGAN-QAOA as a two-stage decision model.

The remainder of this paper is organized as follows.
Section~\ref{sec:two_stage_sp} reviews the general formulation of two-stage stochastic programming and summarizes scenario-set discretization and the associated computational challenges.
Section~\ref{sec:method} describes the qGAN-based scenario-generation circuit and the quantum-circuit implementation of two-stage stochastic programming via qGAN-QAOA.
Sections~\ref{sec:problem_setting} and~\ref{sec:result} present the formulation of the stochastic unit commitment problem with uncertainty in PV output and the corresponding numerical results.
Section~\ref{sec:discussion} discusses the WHT-based Pauli expansion and circuit scaling, and also analyzes the conditions under which the proposed method can be advantageous in computational complexity.
Finally, Section~\ref{sec:conclusion} concludes the paper and outlines future work.

\section{Two-stage stochastic programming}
\label{sec:two_stage_sp}

In this section, we formulate two-stage stochastic programming, which provides the basis for the problem considered in this study, in a general mathematical form. We also summarize the discretization error introduced by scenario-based approximation using SAA and the computational bottleneck of the L-shaped method based on Benders decomposition \cite{L-Shaped,BENDERS}. Throughout this paper, scalar quantities are denoted by non-bold symbols, whereas vector quantities (decision variables and coefficient vectors) are denoted by bold symbols (e.g., $\boldsymbol{x},\boldsymbol{y},\boldsymbol{c}$).

\subsection{General formulation of two-stage stochastic programming}
\label{subsec:General_formulation_two_stage_sp}

As uncertainty, we consider a one-dimensional random variable $\xi:\Omega \to \mathbb{R}$ defined on a probability space $(\Omega,\mathcal{F},\mathbb{P})$. The decision maker chooses the first-stage decision variable $\boldsymbol{x}\in\mathbb{R}^{n}$ only once before observing a realization of $\xi$, and then adjusts the second-stage decision variable $\boldsymbol{y}(\xi)\in\mathbb{R}^{m}$ according to the observed realization $\xi$. In two-stage stochastic programming, the non-anticipativity constraints---namely, the requirement that the first-stage decision variable $\boldsymbol{x}$ be scenario-independent---are imposed \cite{Dantzig_sp,shapiro2009lectures,birge2011intro}.

Let the first-stage cost be defined as $f(\boldsymbol{x}):=\boldsymbol{c}^{\mathsf{T}}\boldsymbol{x}$, and the second-stage recourse cost as $\boldsymbol{q}(\xi)^{\mathsf{T}}\boldsymbol{y}(\xi)$. Here, $\boldsymbol{c}\in\mathbb{R}^{n}$ is the first-stage cost coefficient vector, and $\boldsymbol{q}(\xi)\in\mathbb{R}^{m}$ is the second-stage cost coefficient that depends on the uncertainty $\xi$. Suppose further that the second-stage constraints are given by
\begin{equation}
  W(\xi)\,\boldsymbol{y}
    = \boldsymbol{h}(\xi) - T(\xi)\,\boldsymbol{x},
  \qquad
  \boldsymbol{y} \ge \boldsymbol{0}
  \label{eq:second_stage_constraints}
\end{equation}
where $W(\xi)$ and $T(\xi)$ are matrices of appropriate dimensions and $\boldsymbol{h}(\xi)$ is a vector function; all of these represent problem parameters that depend on $\xi$.

Given a first-stage solution $\boldsymbol{x}$ and a realization $\xi$, we define the optimal second-stage recourse function $Q(\boldsymbol{x},\xi)$ as
\begin{equation}\label{eq:second_stage_general}
\begin{aligned}
  Q(\boldsymbol{x},\xi)
  :=
  \min_{\boldsymbol{y} \in \mathbb{R}^m}
  \left\{
    \boldsymbol{q}(\xi)^\mathsf{T}\boldsymbol{y}
    \;\middle|\;
    W(\xi)\,\boldsymbol{y}
    = \boldsymbol{h}(\xi) - T(\xi)\,\boldsymbol{x},\;
    \boldsymbol{y} \ge \boldsymbol{0}
  \right\}.
\end{aligned}
\end{equation}
Using the recourse function, the two-stage stochastic programming problem is generally written as
\begin{equation}\label{eq:two_stage_sp_general_obj}
\begin{aligned}
  \min_{\boldsymbol{x}\in X}
  \;& \boldsymbol{c}^\mathsf{T}\boldsymbol{x}
       + \mathbb{E}_{\xi}\!\left[
           Q(\boldsymbol{x},\xi)
         \right],\\
  \text{s.t.}\;\;
  & A\boldsymbol{x} \ge \boldsymbol{b},
\end{aligned}
\end{equation}
where $X := \{\boldsymbol{x}\in\mathbb{R}^n \mid A\boldsymbol{x}\ge \boldsymbol{b}\}$ denotes the convex feasible set of the first stage, and $\mathbb{E}_{\xi}[\cdot]$ denotes the expectation with respect to $\xi$.

\subsection{Scenario discretization and computational bottlenecks}
\label{subsec:scenario_discretization}

In practice, it is often difficult to explicitly identify the population distribution of the uncertainty $\xi$. In such cases, one typically relies on samples of $\xi$---e.g., observational data and, when needed, samples generated by a learned generative model---and approximates the expectation $\mathbb{E}_{\xi}[Q(\boldsymbol{x},\xi)]$ by a sample average within the SAA framework \cite{shapiro2009lectures,birge2011intro}. 
Accordingly, we generate i.i.d.\ samples $\xi_0,\dots,\xi_{N-1}$ 
and approximate the expectation by SAA as
\begin{equation}
  \mathbb{E}_{\xi}\!\left[Q(\boldsymbol{x},\xi)\right]
  \approx
  \frac{1}{N}\sum_{s=0}^{N-1} Q(\boldsymbol{x},\xi_s)
  \label{eq:two_stage_sp_discrete}
\end{equation}
Under the standard assumption that $Q(\boldsymbol{x},\xi)$ has finite variance, the statistical error of the sample average scales as
\begin{equation}
  \varepsilon_{\mathrm{est}} = O(N^{-1/2})
  \quad\Rightarrow\quad
  N = O(\varepsilon_{\mathrm{est}}^{-2})
  \label{eq:mc_rate}
\end{equation}
Therefore, evaluating the expected recourse cost with accuracy $\varepsilon_{\mathrm{est}}$ requires $N=O(\varepsilon_{\mathrm{est}}^{-2})$ scenarios.

On the other hand, as $N$ increases, evaluating the approximation in Eq.~\eqref{eq:two_stage_sp_discrete} generally requires solving the second-stage problem~\eqref{eq:second_stage_general} for each scenario $s$ and aggregating the results. In the L-shaped method based on Benders decomposition \cite{BENDERS,L-Shaped}, if the computational cost of solving the second-stage problem once is $T_{\mathrm{2nd}}$, then the computational complexity per iteration is lower bounded by
$O\bigl(N\,T_{\mathrm{2nd}}\bigr)$.
Consequently, increasing the number of scenarios $N$ to reduce the statistical error directly leads to higher computational complexity in two-stage stochastic programming \cite{shapiro2009lectures,birge2011intro}.

\subsection{Motivation and positioning}
\label{subsec:Motivation_positioning}
As discussed above, classical SAA-based two-stage stochastic programming evaluates the expected recourse cost by scenario enumeration, resulting in computational bottlenecks that scale at least linearly with the number of scenarios $N$ (e.g., $O(NT_{\mathrm{2nd}})$ per iteration in the L-shaped method). In this paper, we aim to mitigate the explicit dependence on $N$ by developing qGAN-QAOA, a unified quantum-circuit workflow that encodes the scenario distribution into a quantum state and evaluates the two-stage objective including expected recourse cost as the expectation value of a problem Hamiltonian, optimizing only the QAOA variational parameters.

\section{Method}
\label{sec:method}

\begin{sidewaysfigure*}[p] 
\centering
\includegraphics[width=\textheight,height=\textwidth,keepaspectratio]{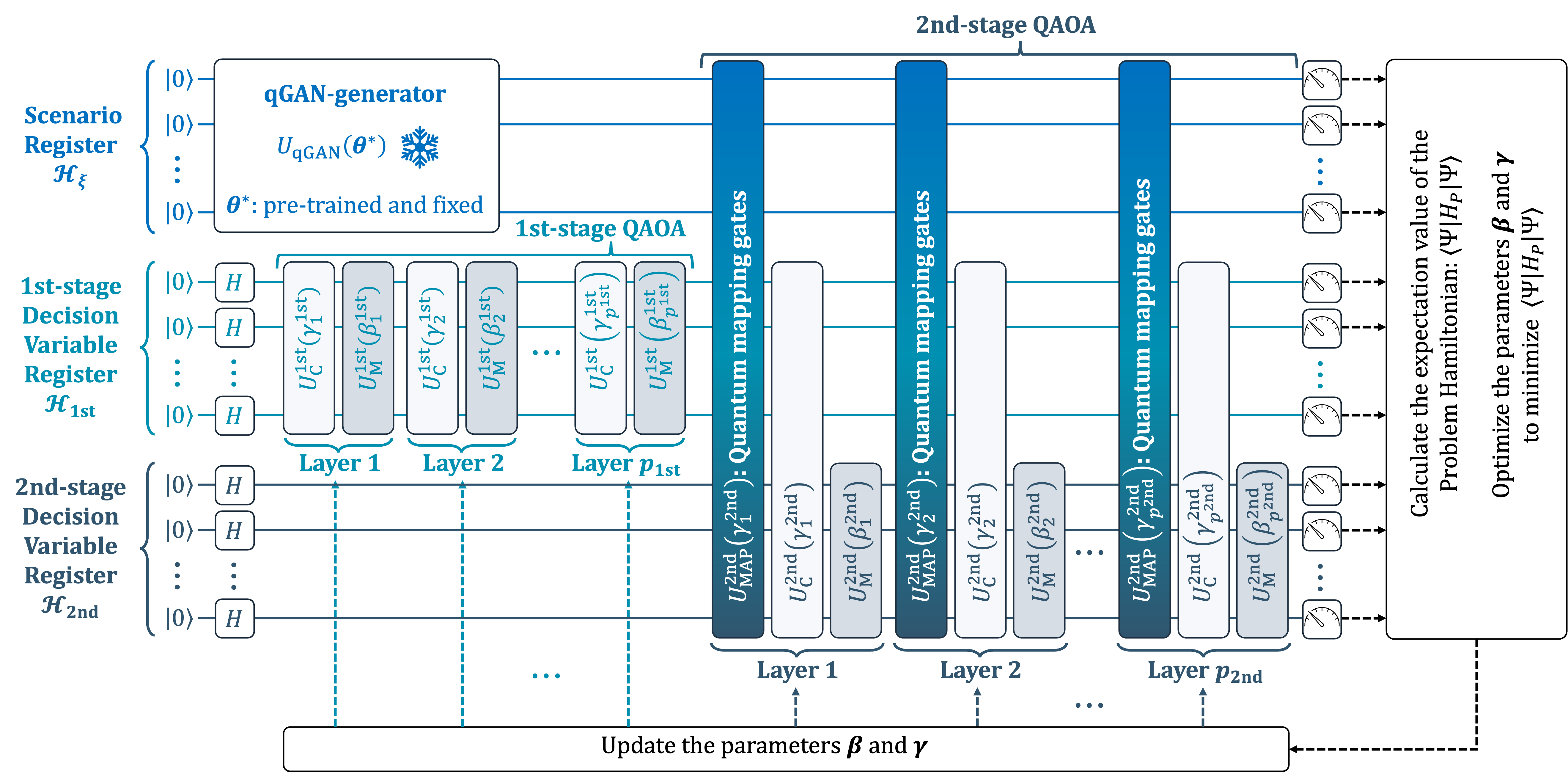}
\caption{
qGAN-QAOA workflow for two-stage stochastic programming.
A pre-trained and fixed qGAN generator $U_{\mathrm{qGAN}}(\boldsymbol{\theta}^{\ast})$ prepares on the scenario register $\mathcal{H}_{\xi}$ a state whose amplitudes encode the scenario distribution,
$|\psi_{\xi}\rangle=U_{\mathrm{qGAN}}(\boldsymbol{\theta}^{\ast})|0\rangle^{\otimes n_{\xi}}$.
In the first-stage QAOA with $p_{\mathrm{1st}}$ layers, the cost unitary $U_{C}^{\mathrm{1st}}(\boldsymbol{\gamma}^{\mathrm{1st}})$ and the mixer unitary $U_{M}^{\mathrm{1st}}(\boldsymbol{\beta}^{\mathrm{1st}})$ are applied alternately to the first-stage register $\mathcal{H}_{\mathrm{1st}}$.
In the second-stage QAOA with $p_{\mathrm{2nd}}$ layers, a quantum mapping gate $U_{\mathrm{MAP}}(\gamma_{\ell}^{\mathrm{2nd}})$ is inserted at the beginning of each layer to embed scenario-conditioned second-stage cost terms into the second-stage register $\mathcal{H}_{\mathrm{2nd}}$, followed by the cost unitary $U_{C}^{\mathrm{2nd}}(\boldsymbol{\gamma}^{\mathrm{2nd}})$ and the mixer unitary $U_{M}^{\mathrm{2nd}}(\boldsymbol{\beta}^{\mathrm{2nd}})$.
Finally, all registers are measured to estimate the problem-Hamiltonian expectation value $\langle H_{P}\rangle$, and the parameters
$\{\boldsymbol{\beta}^{\mathrm{1st}},\boldsymbol{\gamma}^{\mathrm{1st}},\boldsymbol{\beta}^{\mathrm{2nd}},\boldsymbol{\gamma}^{\mathrm{2nd}}\}$
are updated via classical optimization (dashed lines indicate classical feedback).
}
\label{fig:qGAN_QAOA_flow}
\end{sidewaysfigure*}
In this section, we propose qGAN-QAOA, a method for implementing the discrete-scenario two-stage stochastic programming problem defined in Sec.~\ref{sec:two_stage_sp} as a quantum circuit-based framework, as illustrated in Fig.~\ref{fig:qGAN_QAOA_flow}. 
The proposed method uses three quantum registers,
$\mathcal{H}=\mathcal{H}_{2\mathrm{nd}}\otimes\mathcal{H}_{1\mathrm{st}}\otimes\mathcal{H}_{\xi}$.
The register $\mathcal{H}_{\xi}$ is the scenario register, whose amplitudes encode the scenario(uncertainty) distribution $\{p_s\}$. The registers $\mathcal{H}_{1\mathrm{st}}$ and $\mathcal{H}_{2\mathrm{nd}}$ are the first-stage and second-stage decision-variable registers, respectively, and are optimized variationally via QAOA.
The key ideas of the proposed method are summarized as follows: (i) preparing the uncertainty distribution as a quantum state, and (ii) constructing a gate (a quantum mapping gate) that incorporates scenario information into the second-stage cost. 
These enable the two-stage objective including expected recourse cost to be evaluated as the expectation value of a problem Hamiltonian and minimized within a unified quantum-circuit workflow.

\subsection{Encoding the scenario distribution using qGAN}
\label{subsec:Quantization_RV}

\begin{figure}[t] 
\centering \includegraphics[width=\linewidth, keepaspectratio]{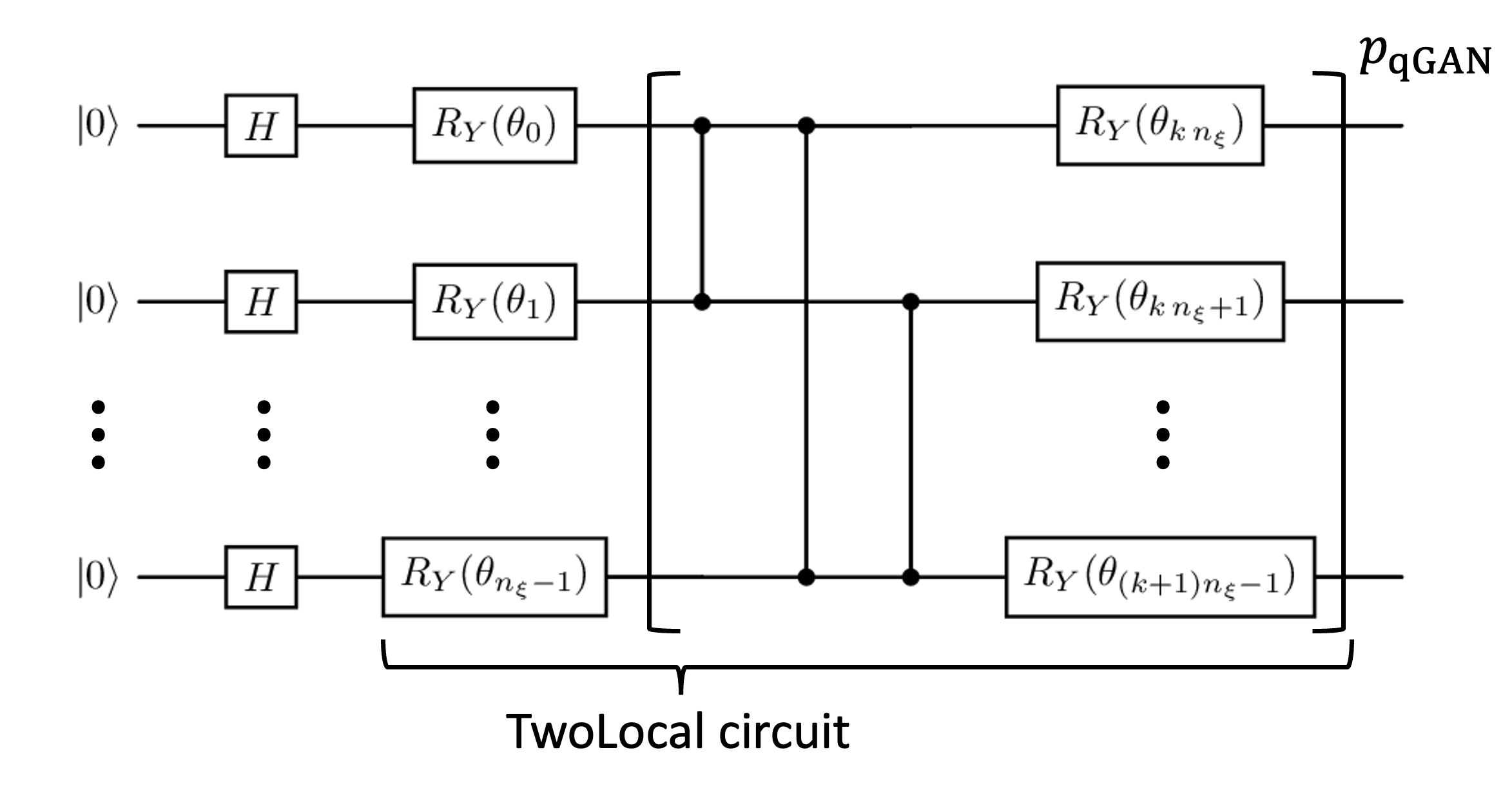} 
\caption{
Quantum generator circuit $U_{\mathrm{qGAN}}(\boldsymbol{\theta})$ (TwoLocal circuit) used in qGAN.
After applying an $H$ gate and a parameterized rotation $R_y(\theta)$ to each qubit, we alternate an entangling layer (a sequence of controlled gates) with an $R_y(\theta)$ layer.
By stacking these layers, the circuit generates a quantum state whose amplitudes represent the scenario distribution parameterized by $\boldsymbol{\theta}$.
}
\label{fig:qGAN_qc} 
\end{figure}

\begin{figure}[t] 
\centering 
\includegraphics[width=\linewidth, keepaspectratio]{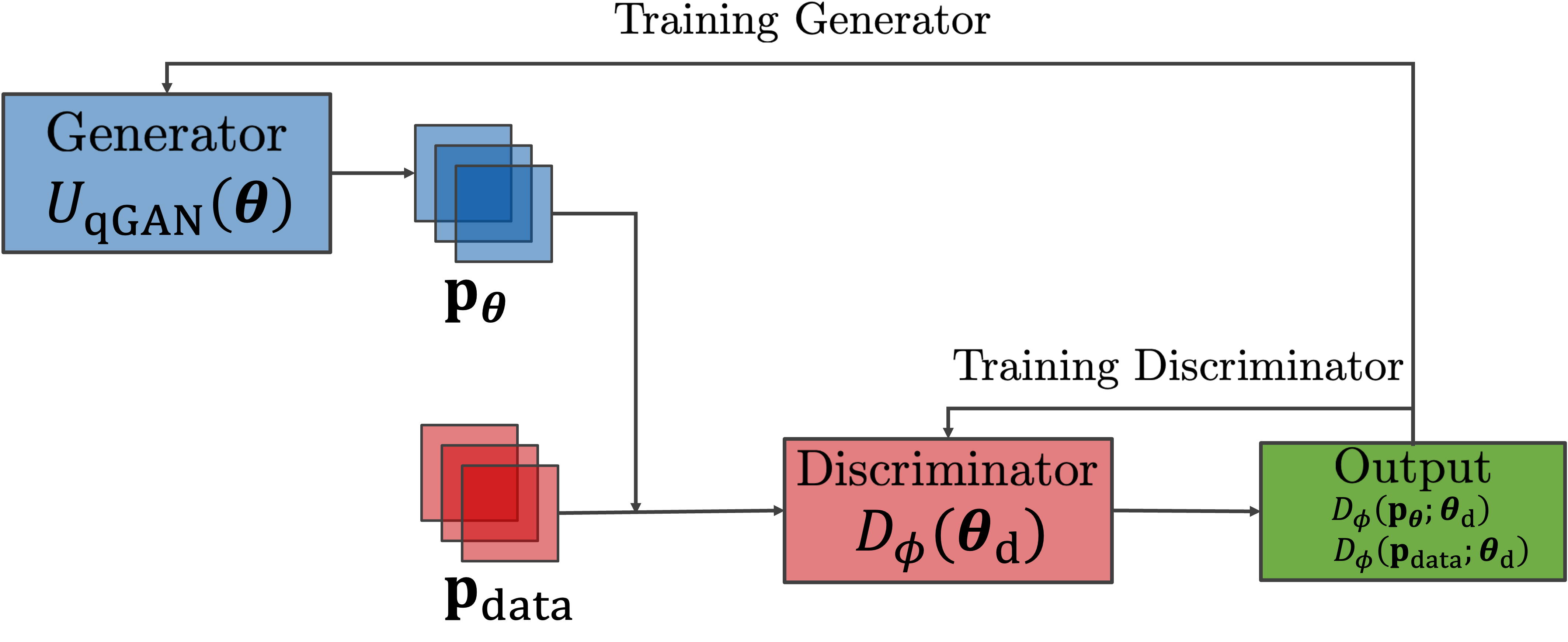} 
\caption{
Training procedure of qGAN.
The quantum generator produces a probability distribution via the parameterized circuit $U_{\mathrm{qGAN}}(\boldsymbol{\theta})$, yielding the generated probability vector $\mathbf{p}_{\boldsymbol{\theta}}$.
The classical discriminator $D_{\boldsymbol{\phi}}$ takes as input either the target distribution $\mathbf{p}_{\mathrm{data}}$ or the generated distribution $\mathbf{p}_{\boldsymbol{\theta}}$ and outputs the probability that the input originates from real data.
Gradients of the loss are propagated to update the parameters $\bm{\theta},\bm{\theta}_{\mathrm{d}}$ iteratively.
} 
\label{fig:qGAN_train_flow} 
\end{figure} 

In this paper, we treat the uncertainty $\xi$ as a one-dimensional continuous random variable on an interval $[\xi_{\min},\xi_{\max}]$. In the proposed method, we learn a discrete distribution of $\xi$ from observational data and prepare a scenario-superposition state that follows the learned distribution using a quantum generative adversarial network (qGAN)\cite{Zoufal_2019_qGAN}. Specifically, we employ a parameterized quantum circuit $U_{\mathrm{qGAN}}(\boldsymbol{\theta})$ acting on $n_\xi$ qubits as the quantum generator. In this paper, we adopt the TwoLocal circuit shown in Fig.~\ref{fig:qGAN_qc}, which consists of rotation gates $R_y$ and entanglement gates $CZ$ (controlled-$Z$). Moreover, to increase the expressive power with the number of qubits $n_\xi$ while keeping the gate count of the qGAN circuit polynomial in $n_\xi$, we set the repetition number to $p_{\mathrm{qGAN}}=n_{\xi}$.
The discriminator is implemented as a classical neural network and is defined as a mapping
$D_{\boldsymbol{\phi}}:\ \mathbb{R}^N \to (0,1)$
which outputs the probability that its input distribution originates from real data. Specifically, the input is either the probability vector
$\mathbf{p}_{\boldsymbol{\theta}}=(p_0,\dots,p_{N-1})$
estimated from measurement frequencies of the quantum generator, or the target distribution
$\mathbf{p}_{\mathrm{data}}=(p^{\mathrm{data}}_0,\dots,p^{\mathrm{data}}_{N-1})$
constructed from observational data. The value $D_{\boldsymbol{\phi}}(\mathbf{p})$ thus represents the probability that $\mathbf{p}$ is classified as a real-data distribution (Fig.~\ref{fig:qGAN_train_flow}). Training is performed via adversarial learning by alternately updating the parameters of the generator and the discriminator.

On the other hand, since the proposed method encodes a probability distribution as amplitudes on the scenario register, it is necessary to represent $\xi$ as a discrete distribution supported on finitely many representative points. We therefore define the scenario set as
\begin{equation}\label{eq:scenarios_set}
  S := \{\xi_0,\dots,\xi_{N-1}\}
\end{equation}
and denote by $p_s$ the probability of each $\xi_s\in S$, assuming $\sum_{s=0}^{N-1}p_s=1$. As a standard discretization scheme, we use uniform discretization,
\begin{equation}\label{eq:disc_rv}
  \xi_s = \xi_{\min} + s\,\Delta\xi,
  \qquad
  \Delta\xi = \frac{\xi_{\max}-\xi_{\min}}{N-1}
\end{equation}
In this case, if $Q(\boldsymbol{x},\xi)$ is sufficiently smooth with respect to $\xi$ and the probabilities $p_s$ are chosen according to a standard density-based quadrature rule, the error of a simple quadrature approximation can be bounded as \cite{atkinson1989numana}
\begin{equation}\label{eq:disc_error}
  \varepsilon_{\mathrm{disc}} = O(\Delta\xi)=O(1/N)
\end{equation}

The main reason we adopt uniform discretization in this paper is that, as shown in Sec.~\ref{subsec:WHT_Pauli}, the random-variable operator on the scenario register becomes sparse under the Walsh--Hadamard transform, and this structure can be exploited in circuit design. Owing to this structure, the scaling of circuit resources can be suppressed even as the number of scenarios $N$ increases.

For the target distribution $\{p_s\}_{s=0}^{N-1}$ determined by the above discretization, we obtain the optimized parameters $\boldsymbol{\theta}^\ast$ through adversarial training of qGAN, which yields
\begin{equation}\label{eq:qGAN_state}
\begin{aligned}
  \ket{\psi^{\mathrm{qGAN}}_\xi}
  =
  U_{\mathrm{qGAN}}(\boldsymbol{\theta}^\ast)\ket{0}^{\otimes n_\xi}
  =
  \sum_{s=0}^{N-1}\sqrt{p_{\boldsymbol{\theta}^\ast}(\xi_s)}\,\ket{b_s^\xi}
  \approx
  \sum_{s=0}^{N-1}\sqrt{p_s}\,\ket{b_s^\xi}
  =\ket{\psi_\xi}
\end{aligned}
\end{equation}
Hereafter, we assume that the trained distribution $p_{\boldsymbol{\theta}^\ast}$ approximates the target distribution $\{p_s\}$ sufficiently well and does not constitute a dominant source of error in the expectation evaluation considered in this paper (validation of the training accuracy is provided in Sec.~\ref{subsec:qGAN-acc}). Therefore, we treat the qGAN generator as sufficiently well pre-trained and fix $\boldsymbol{\theta}^\ast$, so that the optimization targets are restricted to the QAOA variational parameters for the first and second stages.

That is, using the qGAN generator, the scenario register $\mathcal{H}_{\xi}$ consists of $n_\xi=\lceil \log N\rceil$ qubits and we can prepare the amplitude-encoded scenario state
\begin{equation}
  \ket{\psi_\xi}
  =
  \sum_{s=0}^{N-1} \sqrt{p_s}\,\ket{b_s^\xi}
  \label{eq:xi_state}
\end{equation}
where $\ket{b_s^\xi}$ denotes the computational basis state corresponding to scenario $s$. Here, We define the random-variable operator (a diagonal operator) obtained by quantizing the random variable as
\begin{equation}
  \hat{\xi}
  :=
  \mathrm{diag}(\xi_0,\xi_1,\dots,\xi_{N-1})
  \label{eq:RandomVariableOperator}
\end{equation}
so that $\hat{\xi}\ket{b_s^\xi}=\xi_s\ket{b_s^\xi}$. Hereafter, for simplicity, we assume $n_\xi=\log N$.

\subsection{Quantum circuit realization of the two-stage stochastic programming}
\label{subsec:quantum_objective}

In this subsection, we construct the qGAN-QAOA circuit shown in Fig.~\ref{fig:qGAN_QAOA_flow} and show that the objective of two-stage stochastic programming, including expected recourse cost, can be realized as a circuit expectation value. The key point is that, by coherently entangling the three registers---the first-stage, second-stage, and scenario registers---the term $\sum_s p_s Q(\boldsymbol{x},\xi_s)$ can be evaluated as the expectation value of a single quantum circuit.

\subsubsection{First-stage formulation}
\label{subsubsec:1st_stage_form}

We encode the first-stage decision variable $\boldsymbol{x}\in\mathbb{R}^n$ into the register $\mathcal{H}_{1\mathrm{st}}$ consisting of $n_{1\mathrm{st}}=\sum_i n_{1\mathrm{st},i}$ qubits. Each component $x_i$ is represented using an $n_{1\mathrm{st},i}$-bit binary fixed-point encoding as
\begin{equation}
  x_i
  =
  x_i^{\min}
  +
  \Delta_i^{1\mathrm{st}}
  \sum_{j=1}^{n_{1\mathrm{st},i}} 2^{j-1}\, b^{1\mathrm{st}}_{i,j},
  \quad
  \Delta_i^{1\mathrm{st}}=\frac{x_i^{\max}-x_i^{\min}}{2^{n_{1\mathrm{st},i}}-1}
  \label{eq:x_encoding_1st}
\end{equation}
We further quantize the binary variables by $b^{1\mathrm{st}}_{i,j}=(1-\hat{Z}^{1\mathrm{st}}_{i,j})/2$ to obtain the first-stage decision-variable operator $\hat{\boldsymbol{x}}^{1\mathrm{st}}$. Then,
$\hat{\boldsymbol{x}}^{1\mathrm{st}}\ket{b_k^{1\mathrm{st}}}
=\boldsymbol{x}^{1\mathrm{st}}_k\ket{b_k^{1\mathrm{st}}}$
holds.

In the proposed method, we define the first-stage cost Hamiltonian based on the term $f(\boldsymbol{x}):=\boldsymbol{c}^{\mathsf{T}}\boldsymbol{x}$, excluding the expected recourse-cost term, as
\begin{equation}
  H_P^{1\mathrm{st}}
  =
  \sum_i c_i\,\hat{x}_i^{1\mathrm{st}}
  \label{eq:HP1}
\end{equation}
and let the first-stage mixer act only on $\mathcal{H}_{1\mathrm{st}}$.
That is, we define the mixer Hamiltonian as
\begin{equation}
  H_M^{1\mathrm{st}}
  =
  -\sum_{i=0}^{n_{1\mathrm{st}}-1}
  \hat{I}_{2\mathrm{nd}}\otimes \hat{X}^{1\mathrm{st}}_{i}\otimes \hat{I}_{\xi}
  \label{eq:HM1}
\end{equation}
and, using a $p_{1\mathrm{st}}$-layer QAOA, obtain
\begin{equation}\label{eq:x1_state}
\begin{aligned}
  \ket{x^{1\mathrm{st}}}
  &=
  \prod_{\ell=1}^{p_{1\mathrm{st}}}
    \Bigl[
    U_M^{1\mathrm{st}}(\beta_\ell^{1\mathrm{st}})
    U_C^{1\mathrm{st}}(\gamma_\ell^{1\mathrm{st}})
    \ket{+}^{\otimes n_{1\mathrm{st}}}
    \Bigl]\\
  &=
  \sum_{k=0}^{2^{n_{1\mathrm{st}}}-1}
    \alpha_k\,\ket{b_k^{1\mathrm{st}}}
\end{aligned}
\end{equation}
Here,
$U_C^{1\mathrm{st}}(\gamma)=\exp(-i\gamma H_P^{1\mathrm{st}})$ and
$U_M^{1\mathrm{st}}(\beta)=\exp(-i\beta H_M^{1\mathrm{st}})$, which correspond to the first-stage QAOA block in Fig.~\ref{fig:qGAN_QAOA_flow}.
That is, by executing the first-stage variational circuit before the second-stage variational circuit, the circuit architecture ensures a property corresponding to non-anticipativity constraints: the first-stage measurement outcome statistics do not change when conditioned on the scenario (Sec.~\ref{subsec:non_anticipativity}).

\subsubsection{Second-stage formulation}
\label{subsubsec:2nd_stage_form}

In classical formulations, the second-stage variables $\boldsymbol{y}(\xi_s)$ appear as distinct variables for each scenario. In the proposed method, however, scenarios are handled simultaneously in superposition on a single register $\mathcal{H}_{2\mathrm{nd}}$. This is also consistent with existing frameworks that parallelize the second-stage problem using scenario superposition \cite{rotello2024_expected_value}. Using $\mathcal{H}_{2\mathrm{nd}}$ consisting of $n_{2\mathrm{nd}}=\sum_i n_{2\mathrm{nd},i}$ qubits, we encode each component in an $n_{2\mathrm{nd},i}$-bit binary fixed-point encoding as
\begin{equation}
  y_i
  =
  y_i^{\min}
  +
  \Delta_i^{2\mathrm{nd}}
  \sum_{j=1}^{n_{2\mathrm{nd},i}} 2^{j-1}\, b^{2\mathrm{nd}}_{i,j},
  \quad
  \Delta_i^{2\mathrm{nd}}=\frac{y_i^{\max}-y_i^{\min}}{2^{n_{2\mathrm{nd},i}}-1}
  \label{eq:y_encoding_2nd}
\end{equation}
and quantize the binary variables by $b^{2\mathrm{nd}}_{i,j}=(1-\hat{Z}^{2\mathrm{nd}}_{i,j})/2$.
We decompose the Hamiltonian including the second-stage cost and constraint penalties as
\begin{equation}\label{eq:HP2_split}
\begin{aligned}
  H_P^{2\mathrm{nd}}
  =
  H^{2\mathrm{nd}}(\hat{\boldsymbol{y}}^{2\mathrm{nd}},\hat{\boldsymbol{x}}^{1\mathrm{st}})
  +
  H^{2\mathrm{nd}}(\hat{\boldsymbol{y}}^{2\mathrm{nd}},\hat{\boldsymbol{x}}^{1\mathrm{st}},\hat{\xi})    
\end{aligned}
\end{equation}
where the first term is scenario-independent and the second term is scenario-dependent through $\hat{\xi}$. Since we consider (diagonal) Hamiltonians composed only of Pauli-$Z$ operators, these terms commute. The second-stage mixer is defined as
\begin{equation}
  H_M^{2\mathrm{nd}}
  =
  - \sum_{i=0}^{n_{2\mathrm{nd}}-1}
    \hat{X}^{2\mathrm{nd}}_i \otimes \hat{I}_{1\mathrm{st}+\xi}
  \label{eq:HM2}
\end{equation}
and acts only on $\mathcal{H}_{2\mathrm{nd}}$.

\subsubsection{Second-stage QAOA implementation with quantum mapping gate}
\label{subsubsec:qaoa_mapping_gate}
To apply second-stage QAOA, we define the input state as
\begin{equation}
  \ket{\psi_{\mathrm{in}}}
  :=
  \ket{+}^{\otimes n_{2\mathrm{nd}}}\otimes
  \ket{x^{1\mathrm{st}}}\otimes
  \ket{\psi_\xi}
  \label{eq:init_state_all}
\end{equation}
where $\ket{x^{1\mathrm{st}}}$ is the state prepared by the first-stage QAOA in Sec.~\ref{subsubsec:1st_stage_form} and $\ket{\psi_\xi}$ is the scenario-superposition state constructed in Sec.~\ref{subsec:Quantization_RV}. For notational convenience, the following derivation treats $\ket{x^{1\mathrm{st}}}$ as given; in the actual optimization, however, the first- and second-stage variational parameters are updated jointly, and the overall procedure minimizes a single objective function (Eq.~\eqref{eq:loss_def}).

We implement the second stage using a $p_{2\mathrm{nd}}$-layer QAOA. Specifically, for each layer $\ell$ we apply
\begin{equation}\label{eq:2nd_layer}
\begin{aligned}
  U^{2\mathrm{nd}}_\ell
  &:=
  U^{2\mathrm{nd}}_M(\beta_\ell^{2\mathrm{nd}})
  \,U^{2\mathrm{nd}}_P(\gamma_\ell^{2\mathrm{nd}})\\
  &=
  \exp\!\bigl(-i\beta_\ell^{2\mathrm{nd}} H_M^{2\mathrm{nd}}\bigr)\,
  \exp\!\bigl(-i\gamma_\ell^{2\mathrm{nd}} H_P^{2\mathrm{nd}}\bigr)
\end{aligned}
\end{equation}
Here, $H_M^{2\mathrm{nd}}$ is a mixer acting only on the second-stage decision-variable register $\mathcal{H}_{2\mathrm{nd}}$, and $H_P^{2\mathrm{nd}}$ is a diagonal Hamiltonian representing the second-stage cost and constraint penalties. In this paper, we decompose $H_P^{2\mathrm{nd}}$ into a term independent of uncertainty and a term depending on uncertainty (and the first stage) as
\begin{equation}
  H_P^{2\mathrm{nd}}
  =
  H^{2\mathrm{nd}}(\hat{\boldsymbol{y}}^{2\mathrm{nd}},\hat{\boldsymbol{x}}^{1\mathrm{st}})
  +
  H^{2\mathrm{nd}}(\hat{\boldsymbol{y}}^{2\mathrm{nd}},\hat{\boldsymbol{x}}^{1\mathrm{st}},\hat{\xi})
  \label{eq:Hp2_decomp_here}
\end{equation}
(which is equivalent to Eq.~\eqref{eq:HP2_split}). Since both terms are diagonal operators composed only of Pauli-$Z$ operators, they commute, and thus
\begin{equation}
  \exp\!\bigl(-i\gamma_\ell^{2\mathrm{nd}} H_P^{2\mathrm{nd}}\bigr)
  =
  U_C^{2\mathrm{nd}}(\gamma_\ell^{2\mathrm{nd}})
  \,U_{\mathrm{MAP}}^{2\mathrm{nd}}(\gamma_\ell^{2\mathrm{nd}})
  \label{eq:split_map}
\end{equation}
where
\begin{align}
  U_C^{2\mathrm{nd}}(\gamma)
  &:=
  \exp\!\Bigl(
    -i\gamma\,
    H^{2\mathrm{nd}}(\hat{\boldsymbol{y}}^{2\mathrm{nd}},\hat{\boldsymbol{x}}^{1\mathrm{st}})
  \Bigr)
  \label{eq:Uc2_def}\\
  U_{\mathrm{MAP}}^{2\mathrm{nd}}(\gamma)
  &:=
  \exp\!\Bigl(
    -i\gamma\,
    H^{2\mathrm{nd}}(\hat{\boldsymbol{y}}^{2\mathrm{nd}},\hat{\boldsymbol{x}}^{1\mathrm{st}},\hat{\xi})
  \Bigr)
  \label{eq:Umap2_def}
\end{align}
In this paper, we refer to $U_{\mathrm{MAP}}^{2\mathrm{nd}}$ as the quantum mapping gate. This gate acts on all registers (scenario, first-stage, and second-stage) rather than only on the second-stage register, and implements a phase rotation of the cost-Hamiltonian gate corresponding to the second-stage cost term
$H^{2\mathrm{nd}}(\hat{\boldsymbol{y}}^{2\mathrm{nd}},\boldsymbol{x}_k^{1\mathrm{st}},\xi_s)$
conditioned on the computational basis state $\ket{b_k^{1\mathrm{st}}}\ket{b_s^\xi}$, thereby enabling a conditional evaluation of the costs for all $(k,s)$ within the same circuit.

Consequently, the final state can be written as
\begin{align}
  \ket{\Psi}
  &=
  \prod_{\ell=1}^{p_{2\mathrm{nd}}}
  \Bigl[
    U_M^{2\mathrm{nd}}(\beta_\ell^{2\mathrm{nd}})
    U_C^{2\mathrm{nd}}(\gamma_\ell^{2\mathrm{nd}})
    U_{\mathrm{MAP}}^{2\mathrm{nd}}(\gamma_\ell^{2\mathrm{nd}})
  \Bigr]
  \ket{\psi_{\mathrm{in}}}
  \nonumber\\
  &=
  \sum_{k=0}^{2^{n_{1\mathrm{st}}}-1}
  \alpha_k
  \sum_{s=0}^{N-1} \sqrt{p_s}\,
    \ket{\psi_{s,k}^{2\mathrm{nd}}}
    \otimes \ket{b_k^{1\mathrm{st}}}
    \otimes \ket{b_s^\xi}
  \label{eq:psi_final_expanded}
\end{align}
where
\begin{equation}\label{eq:2nd_final_state}
\begin{aligned}
  \ket{\psi_{s,k}^{2\mathrm{nd}}}
  :=
  \prod_{\ell=1}^{p_{2\mathrm{nd}}}
  \Bigl[
    U_M^{2\mathrm{nd}}(\beta_\ell^{2\mathrm{nd}})
    \exp\Bigl(
      -i\gamma_\ell^{2\mathrm{nd}}\,
      H_P^{2\mathrm{nd}}\bigl(
        \hat{\boldsymbol{y}}^{2\mathrm{nd}},
        \boldsymbol{x}_k^{1\mathrm{st}},
        \xi_s
      \bigr)
    \Bigr)
  \Bigr]
  \ket{+}^{\otimes n_{2\mathrm{nd}}}
\end{aligned}
\end{equation}
is the final state of the second-stage QAOA for first-stage solution $\boldsymbol{x}_k^{1\mathrm{st}}$ and scenario $\xi_s$. Since it is expressed solely by unitary gates acting on the register $\mathcal{H}_{2\mathrm{nd}}$, we have $\langle\psi_{s,k}^{2\mathrm{nd}}|\psi_{s,k}^{2\mathrm{nd}}\rangle=1$. Therefore, Eq.~\eqref{eq:psi_final_expanded} provides an ansatz that simultaneously prepares the second-stage states corresponding to each $(k,s)$ under the first-stage distribution $|\alpha_k|^2$ and the scenario distribution $p_s$. Note that, while $U_{\mathrm{MAP}}^{2\mathrm{nd}}$ introduces a conditional phase depending on $(k,s)$, it does not make the marginal measurement distribution of the first-stage register depend on the scenario, as shown in Sec.~\ref{subsec:non_anticipativity}.

\subsubsection{Quantum objective and variational optimization}
\label{subsubsec:quantum_objective_opt}

We define the overall problem Hamiltonian as
\begin{equation}
  H_P
  :=
  H_P^{1\mathrm{st}}(\hat{\boldsymbol{x}}^{1\mathrm{st}})
  +
  H_P^{2\mathrm{nd}}(\hat{\boldsymbol{y}}^{2\mathrm{nd}},\hat{\boldsymbol{x}}^{1\mathrm{st}},\hat{\xi})
  \label{eq:total_HP}
\end{equation}
and evaluate the expectation value $\bra{\Psi}H_P\ket{\Psi}$.

\begin{proposition}
  \label{prop:quantum_two_stage_objective}
  Let $\ket{\Psi}$ be defined by Eq.~\eqref{eq:psi_final_expanded}. Then,
  \begin{equation}\label{eq:quantum_two_stage_obj}
  \begin{aligned}
    \bra{\Psi} H_P \ket{\Psi}
    =
    \sum_{k=0}^{2^{n_{1\mathrm{st}}}-1}
    |\alpha_k|^2
    \left[
      H_P^{1\mathrm{st}}(\boldsymbol{x}_k^{1\mathrm{st}})
      +
      \sum_{s=0}^{N-1} p_s\,
      Q(\boldsymbol{x}_k^{1\mathrm{st}},\xi_s)
    \right]
  \end{aligned}
  \end{equation}
  where
  \begin{equation}
    Q(\boldsymbol{x}_k^{1\mathrm{st}},\xi_s)
    :=
    \bra{\psi_{s,k}^{2\mathrm{nd}}}
      H_P^{2\mathrm{nd}}(\hat{\boldsymbol{y}}^{2\mathrm{nd}},\boldsymbol{x}_k^{1\mathrm{st}},\xi_s)
    \ket{\psi_{s,k}^{2\mathrm{nd}}}
    \label{eq:Q_quantum_def}
  \end{equation}
  holds.
\end{proposition}

\begin{proof}
  Substituting Eq.~\eqref{eq:psi_final_expanded} into $\ket{\Psi}$ and using the orthogonality of $\{\ket{b_k^{1\mathrm{st}}}\}_k$ and $\{\ket{b_s^\xi}\}_s$, together with $\langle\psi_{s,k}^{2\mathrm{nd}}|\psi_{s,k}^{2\mathrm{nd}}\rangle=1$, yields Eq.~\eqref{eq:quantum_two_stage_obj}. A detailed proof is provided in App.~\ref{app:proof_quantum_two_stage_objective}.
\end{proof}

Proposition~\ref{prop:quantum_two_stage_objective} shows that minimizing the circuit expectation value in qGAN-QAOA variationally approximates and searches over the classical two-stage objective function
$f(\boldsymbol{x})+\sum_s p_s Q(\boldsymbol{x},\xi_s)$
under the distribution $|\alpha_k|^2$ over the first-stage candidate solutions $\boldsymbol{x}_k^{1\mathrm{st}}$.
Accordingly, for the variational parameters
\(
  \boldsymbol{\theta}=
  (\boldsymbol{\gamma}^{1\mathrm{st}},\boldsymbol{\beta}^{1\mathrm{st}},
   \boldsymbol{\gamma}^{2\mathrm{nd}},\boldsymbol{\beta}^{2\mathrm{nd}})
\),
we define
\begin{equation}
  \mathcal{L}(\boldsymbol{\theta})
  :=
  \bra{\Psi(\boldsymbol{\theta})} H_P \ket{\Psi(\boldsymbol{\theta})}
  \label{eq:loss_def}
\end{equation}
and solve
\begin{equation}
  \min_{\boldsymbol{\theta}} \mathcal{L}(\boldsymbol{\theta})
  \label{eq:qgan_qaoa_variational_problem}
\end{equation}
via classical optimization as shown in Fig.\ref{fig:qGAN_QAOA_flow}.

\subsection{Non-anticipativity constraints}
\label{subsec:non_anticipativity}

In this subsection, we interpret non-anticipativity constraints as a condition on measurement outcome statistics and prove that the marginal distribution of the first-stage measurement outcomes is independent of the scenario.

The non-anticipativity constraints in two-stage stochastic programming require that the first-stage decision vector does not depend on the scenario $s$, i.e., for any $s,s'$,
\begin{equation}
  \boldsymbol{x}_s=\boldsymbol{x}_{s'}
  \label{eq:non_anticipativity_classical}
\end{equation}
must hold. The essential meaning of non-anticipativity is that the first-stage decision is uniquely determined before observing the uncertainty and is not altered by the realized scenario.

In QAOA, solutions are interpreted through the statistical distribution of computational-basis measurement outcomes. Accordingly, 
Accordingly, we introduce non-anticipativity constraints for QAOA as the invariance of the marginal distribution of the first-stage measurement outcomes with respect to the measurement outcomes of the scenario register.

\begin{requirement}[non-anticipativity constraints in QAOA]\label{req:Quantum_non-anticipativity}
Let the bit string obtained by a computational-basis measurement of the first-stage register $\mathcal{H}_{1\mathrm{st}}$ be the random variable $B^{1\mathrm{st}}\in\{0,1\}^{n_{1\mathrm{st}}}$, and let the measurement outcome of the scenario register $\mathcal{H}_{\xi}$ be $B^{\xi}\in\{0,1\}^{n_{\xi}}$. As non-anticipativity constraints in QAOA, for any $s$ (with $\Pr(B^\xi=b_s^\xi)>0$) and any $k$, we assume that
\begin{equation}
  \Pr\bigl(B^{1\mathrm{st}}=b_k^{1\mathrm{st}}\mid B^\xi=b_s^\xi\bigr)
  =
  \Pr\bigl(B^{1\mathrm{st}}=b_k^{1\mathrm{st}}\bigr)
  \label{eq:non_anticipativity_quantum_postulate}
\end{equation}
holds.
\end{requirement}
Here, the measurement outcomes $B^{1\mathrm{st}}=b_k^{1\mathrm{st}}$ and $B^\xi=b_s^\xi$ correspond uniquely to the first-stage decision vector $\boldsymbol{x}_k^{1\mathrm{st}}$ and the scenario $\xi_s$ via Eqs.~\eqref{eq:x_encoding_1st} and~\eqref{eq:disc_rv}. Therefore, the above requirement is equivalent to requiring that the distribution of the first-stage decision vector $\boldsymbol{x}_k^{1\mathrm{st}}$ does not change when conditioned on the measurement outcome of the scenario register, i.e., for any scenarios $s,s'$,
\begin{equation}
  \Pr\bigl(\boldsymbol{x}^{1\mathrm{st}}
  =\boldsymbol{x}_k^{1\mathrm{st}}\mid \xi=\xi_s\bigr)
  =\Pr\bigl(\boldsymbol{x}^{1\mathrm{st}}
  =\boldsymbol{x}_k^{1\mathrm{st}}\mid \xi=\xi_{s'}\bigr),
  \quad \forall k
\end{equation}
is required to hold.

\begin{proposition}[Satisfaction of Req.~\ref{req:Quantum_non-anticipativity}]\label{prop:qna_satisfied}
The proposed method (qGAN-QAOA) satisfies Req.~\ref{req:Quantum_non-anticipativity}.
\end{proposition}
\begin{proof}
We show that Req.~\ref{req:Quantum_non-anticipativity} holds by evaluating the measurement probabilities induced by a computational-basis measurement of the state $\ket{\Psi}$ in Eq.~\eqref{eq:psi_final_expanded}. Define the projectors for computational-basis measurements as
\begin{equation*}
\begin{aligned}
  \Pi_k^{1\mathrm{st}}
  :=\hat{I}_{2\mathrm{nd}}
  \otimes\ket{b_k^{1\mathrm{st}}}
  \bra{b_k^{1\mathrm{st}}}\otimes\hat{I}_{\xi},
  \quad
  \Pi_s^{\xi}
  :=\hat{I}_{2\mathrm{nd}}
  \otimes\hat{I}_{1\mathrm{st}}\otimes
  \ket{b_s^{\xi}}\bra{b_s^{\xi}} 
\end{aligned}
\end{equation*}
Using the definition of conditional probability $\Pr(A\mid B)=\Pr(A,B)/\Pr(B)$ and $\langle\psi_{s,k}^{2\mathrm{nd}}|\psi_{s,k}^{2\mathrm{nd}}\rangle=1$, we obtain
\begin{align*}
  \Pr\bigl(B^{1\mathrm{st}}=b_k^{1\mathrm{st}}\mid B^\xi=b_s^\xi\bigr)
  &=
  \frac{
    \bra{\Psi}\Pi_k^{1\mathrm{st}} \Pi_s^{\xi}\ket{\Psi}
  }{
    \bra{\Psi} \Pi_s^{\xi}\ket{\Psi}
  }\\
  &=
  \frac{|\alpha_k|^2\,p_s}{p_s}
  =|\alpha_k|^2\\
  &=\bra{\Psi}\Pi_k^{1\mathrm{st}}\ket{\Psi}\\
  &=
  \Pr\bigl(B^{1\mathrm{st}}=b_k^{1\mathrm{st}})
  \qquad
  \quad \forall s,k.
\end{align*}
Hence,
\begin{equation}\label{eq:non_anticipativity_result}
  \Pr\bigl(B^{1\mathrm{st}}=b_k^{1\mathrm{st}}\mid B^\xi=b_s^\xi\bigr)
  =\Pr\bigl(B^{1\mathrm{st}}=b_k^{1\mathrm{st}})
  \qquad \forall s,k
\end{equation}
and Req.~\ref{req:Quantum_non-anticipativity} follows.
\end{proof}

On the other hand, since the second-stage state $\ket{\psi_{s,k}^{2\mathrm{nd}}}$ depends on $(s,k)$ through Eq.~\eqref{eq:2nd_final_state}, quantum entanglement between the first-stage register and the scenario register can, in principle, be generated. However, as Eq.~\eqref{eq:non_anticipativity_result} shows, the property essential for non-anticipativity constraints in QAOA---namely, that \emph{in computational-basis measurements} the first-stage measurement outcomes do not depend on the scenario---is preserved. 
Consequently, the proposed method enables scenario-dependent second-stage recourse optimization while maintaining the first-stage non-anticipativity constraints, thereby allowing the two-stage objective including expected recourse cost to be evaluated and minimized within a unified quantum-circuit workflow.

\section{Case study: two-stage stochastic Unit Commitment Problem under PV uncertainty}
\label{sec:problem_setting}

\begin{figure}[h]
\centering
\includegraphics[width=\linewidth, keepaspectratio]{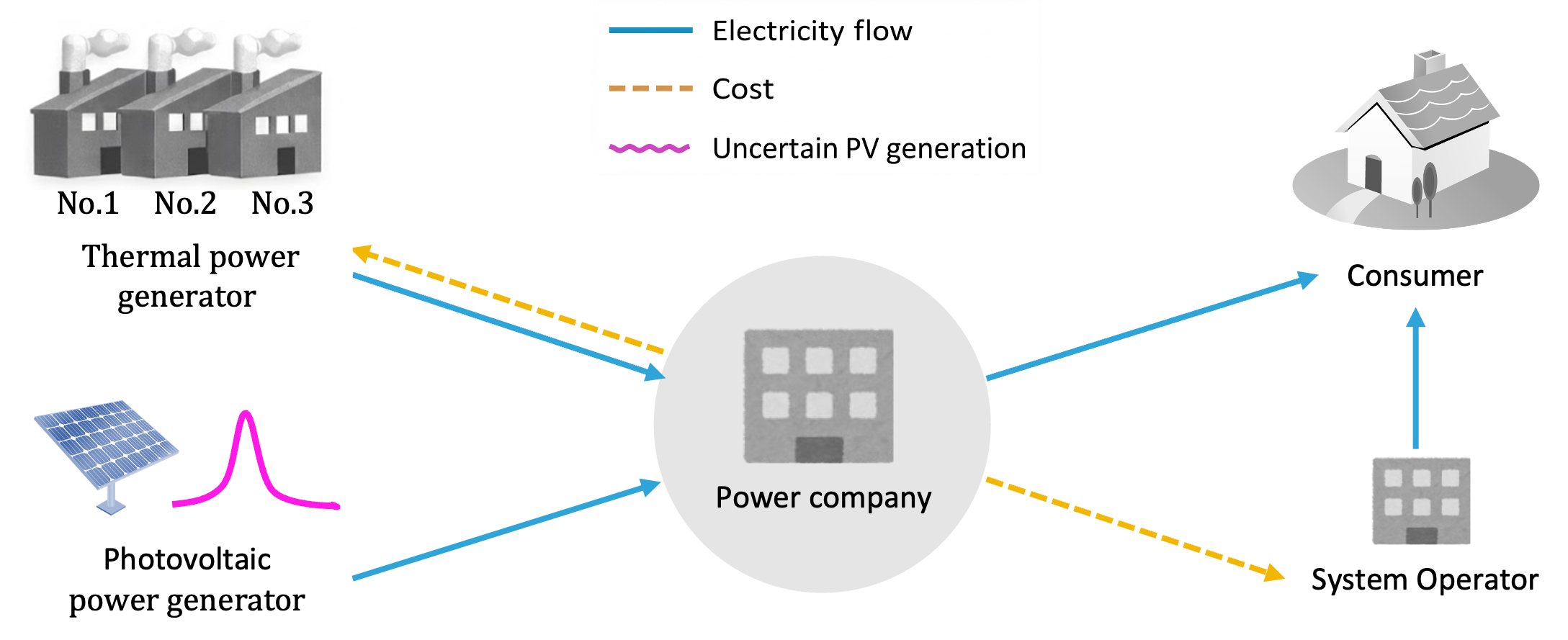}
\caption{
Schematic overview of the stochastic UCP for a power company owning uncertain PV generation.
The power company supplies electricity to consumers by integrating generation from thermal power units (No.~1--3) and PV, while PV output is uncertain (pink), leading to an imbalance-adjustment power flow by the system operator (blue) and an imbalance cost (orange dashed line).
}
\label{fig:s_d_problem}
\end{figure}

In this section, as a concrete case study of the two-stage stochastic programming problem in Sec.~\ref{sec:two_stage_sp}, we formulate a single-period (one-step) stochastic UCP with uncertainty in PV output, as illustrated in Fig.~\ref{fig:s_d_problem}. Moreover, to connect directly to the subsequent quantum-circuit implementation (Sec.~\ref{sec:method}), we provide (i) a mapping to HUBO (higher-order unconstrained binary optimization) and (ii) the explicit form of the corresponding diagonal Hamiltonian. The problem considered here can be regarded as a minimal model of stochastic UCP (start-up and output adjustment) \cite{takriti1996_stochastic_uc,anjos2017_uc,papavasiliou2012_suc_renewables}.

\subsection{Two-stage formulation of stochastic UCP}
\label{subsec:uc_twostage}

We consider a single-period (one-step) setting with $M$ thermal generators whose on/off statuses are decided in the first stage. 
The index $i\in\{1,\dots,M\}$ denotes a thermal generator, and $s\in\{0,\dots,N-1\}$ denotes a PV-output scenario. The demand $D$ is assumed to be a given deterministic value (extensions to include demand uncertainty are possible, but in this section we focus on PV-side uncertainty). For scenario $s$, the PV output is given by the realization $\xi_s$ of the random variable $\xi$, occurring with probability $p_s$. The output range of thermal generator $i$ is $[P^{\mathrm{min}}_i,P^{\mathrm{max}}_i]$, its start-up cost is $d_i$, and its generation cost per unit output is $c_i$. We also introduce an imbalance cost (penalty coefficient) $\lambda>0$ for supply--demand mismatch.

The first-stage decision is a binary variable $x_i\in\{0,1\}$ indicating the on/off status of generator $i$, and we collect them as $\boldsymbol{x}=(x_1,\dots,x_M)^\mathsf{T}$. In the second stage, after observing $\xi_s$, we determine the generator outputs $\boldsymbol{y}_s=(y_{1,s},\dots,y_{M,s})^\mathsf{T}$, where $y_{i,s}\in\mathbb{R}_{\ge 0}$ denotes the output of generator $i$. Let $\sigma_s\in\mathbb{R}$ denote the slack (imbalance) in the power balance, and define the power-balance constraint as
\begin{equation}
  \xi_s+\sum_{i=1}^{M}y_{i,s}+\sigma_s=D
  \label{eq:uc_balance_def}
\end{equation}
Accordingly, the problem is formulated as the following two-stage stochastic programming:
\begin{align}
  \min_{\boldsymbol{x},\{\boldsymbol{y}_s\}}
  \; & \sum_{i=1}^{M} d_i x_i
       + \sum_{s=0}^{N-1} p_s
         \left(
           \sum_{i=1}^{M} c_i y_{i,s}
           + \lambda\,|\sigma_s|
         \right),
  \label{eq:uc_sp_obj_xi}\\
  \text{s.t.}\;\;
  & \sigma_s = D-\xi_s-\sum_{i=1}^{M} y_{i,s},
       \quad \forall s,
  \label{eq:uc_sp_sigma_def}\\
  & P^{\mathrm{min}}_ix_i \le y_{i,s} \le P^{\mathrm{max}}_i x_i,
    \quad \forall i,\ \forall s,
  \label{eq:uc_sp_capacity_y}\\
  & x_i\in\{0,1\},\quad \forall i
  \label{eq:uc_sp_binary_x}
\end{align}
In this two-stage formulation, the first-stage decision $\boldsymbol{x}$ is shared across all scenarios (non-anticipativity constraints), whereas only the second-stage variables $\boldsymbol{y}_s$ depend on $s$, which constitutes the core of the two-stage structure.

\subsection{HUBO formulation using \texorpdfstring{$x$}{x}-controlled binary encoding}

In this subsection, we map Eqs.~\eqref{eq:uc_sp_obj_xi}--\eqref{eq:uc_sp_binary_x} to a HUBO formulation.
From the perspective of quantum implementation, rather than enforcing the conditional constraint $y_{i,s}\le \bar P_i x_i$ via an additional penalty term, it is often simpler to parameterize the feasible set directly at the encoding stage, which can lead to a more concise circuit realization. Accordingly, we represent the second-stage outputs using an \texorpdfstring{$x$}{x}-controlled binary encoding so that Eq.~\eqref{eq:uc_sp_capacity_y} is satisfied structurally without introducing additional penalty terms.

\subsubsection{\texorpdfstring{$x$}{x}-controlled binary encoding of second-stage outputs}
\label{subsubsec:uc_xgated}

We introduce binary variables $b^{2\mathrm{nd}}_{i,s,j}\in\{0,1\}$ ($j=1,\dots,n_{2\mathrm{nd},i}$) and define
\begin{equation}\label{eq:y_gated_encoding}
\begin{aligned}
  &y_{i,s}
  :=
  x_i
  \Bigg(P^{\mathrm{min}}_i +
  \sum_{j=1}^{n_{2\mathrm{nd},i}}
  \Delta^{2\mathrm{nd}}_i\,2^{j-1}\, b^{2\mathrm{nd}}_{i,s,j}
  \Biggl),\quad
  \Delta^{2\mathrm{nd}}_i 
  := \frac{P^{\mathrm{max}}_i-P^{\mathrm{min}}_i}{2^{n_{2\mathrm{nd},i}}-1}
\end{aligned}
\end{equation}
We refer to this construction as an \texorpdfstring{$x$}{x}-controlled binary encoding, since the binary expansion is activated only when $x_i=1$.
With $x_i\in\{0,1\}$, we have $y_{i,s}=0$ whenever $x_i=0$, whereas for $x_i=1$ it follows that $0\le y_{i,s}\le \bar P_i$.
Therefore, Eq.~\eqref{eq:uc_sp_capacity_y} is not removed; rather, it is embedded as an encoding that directly represents a subset of the feasible set.

\subsubsection{Eliminating the imbalance variable and using a quadratic penalty}
\label{subsubsec:uc_sigma_elim}

From Eq.~\eqref{eq:uc_sp_sigma_def}, we can eliminate $\sigma_s$ as
$\sigma_s = D-\xi_s-\sum_{i=1}^{M} y_{i,s}$.
Moreover, to ensure compatibility with the HUBO form, in this section we replace the $L_1$ penalty $\lambda\lvert\sigma_s\rvert$ with the $L_2$ penalty $\lambda\sigma_s^2$, and consider
\begin{equation}\label{eq:HUBO_obj_core}
\begin{aligned}
  \min_{\boldsymbol{x},\{b^{2\mathrm{nd}}_{i,s,j}\}}
  \ \sum_{i=1}^{M} d_i x_i
  + \sum_{s=0}^{N-1} \,
   p_s
  \Bigg[
    \sum_{i=1}^{M} c_i y_{i,s}
    + \lambda\Bigl(D-\xi_s-\sum_{i=1}^{M} y_{i,s}\Bigr)^2
  \Bigg]
\end{aligned}
\end{equation}
as the objective for analysis and implementation. This replacement is not an exact equivalence transformation; rather, it is positioned as a smooth surrogate (penalized) model that strongly suppresses supply--demand imbalance.

In this paper, the stochastic UCP is originally formulated with an $L_1$ imbalance penalty $\lvert\sigma_s\rvert$, which we use as the reference objective for evaluation. However, keeping the $L_1$ term prevents a direct mapping to a HUBO form, and thus hinders the Hamiltonian-based implementation required by our quantum-circuit workflow. We therefore replace $\lvert\sigma_s\rvert$ with a quadratic surrogate and use the resulting $L_2$-penalized objective in Eq.~\eqref{eq:HUBO_obj_core} for the HUBO mapping and for running qGAN-QAOA. In Sec.~\ref{subsec:qganqaoa_results}, specifically in Eqs.~\eqref{eq:Exp_cost_qGAN-QAOA}--\eqref{eq:eev}, the solutions obtained by minimizing this surrogate objective are evaluated and compared using the original $L_1$ two-stage objective; the same $L_1$ objective is also used for the classical baselines (RP, EEV).
In this way, since the main focus of this paper is the quantum-circuit realization of the two-stage structure and the associated scaling discussion, we prioritize circuit simplicity over the exactness of the penalty model%
\footnote{%
To map an $L_1$ penalty exactly into a HUBO form, one needs to introduce auxiliary variables to linearize the absolute value and to convert inequality constraints into an unconstrained form (via additional penalties), which generally increases the number of auxiliary bits and quadratic couplings. In addition, choosing penalty coefficients that satisfy sufficient conditions is separately required. In light of the main focus of this paper (circuit-size analysis and implementation simplification), we adopt in this section a surrogate objective based on an $L_2$ penalty.%
}.

\subsubsection{HUBO form and two-stage interpretation}
\label{subsubsec:uc_HUBO_std}

Substituting Eq.~\eqref{eq:y_gated_encoding} into Eq.~\eqref{eq:HUBO_obj_core}, the objective function can be expanded as a fourth-order polynomial in $x_i$ and $b^{2\mathrm{nd}}_{i,s,j}$. Therefore, the problem can be mapped to the standard form of a HUBO with maximum degree $d=4$:
\begin{equation}
  \min_{\boldsymbol{z}\in\{0,1\}^n}\;
  \sum_{k=0}^{4}\ \sum_{1\le i_1<\cdots<i_k\le n}
  Q^{(k)}_{i_1\cdots i_k}\, z_{i_1} z_{i_2}\cdots z_{i_k}
  \label{eq:hubo_standard_form}
\end{equation}
Here, $\boldsymbol{z}$ is the binary decision vector obtained by concatenating $\{x_i\}$ and $\{b^{2\mathrm{nd}}_{i,s,j}\}$. The coefficients $Q^{(k)}_{i_1\cdots i_k}$ correspond to the $k$th-order terms ($k=0,\dots,4$) obtained by expanding Eq.~\eqref{eq:hubo_standard_form}. Although eliminating $\sigma_s$ and adopting a penalty model simplify the constraint representation, the two-stage interpretation---namely, that the second-stage decisions are selected after observing $\xi_s$---is preserved. In this sense, the HUBO constructed in this section is consistent as a penalized approximation of the two-stage stochastic program.

\subsection{Diagonal Hamiltonian for stochastic UCP in qGAN-QAOA}
\label{subsec:uc_hamiltonian}

In this subsection, we provide only the operator definitions required to map Eq.~\eqref{eq:HUBO_obj_core} to a diagonal Hamiltonian for quantum-circuit implementation (for the circuit construction and the proof that the circuit expectation reproduces the two-stage objective, see Sec.~\ref{sec:method}). Since the scenario probabilities $\{p_s\}$ are amplitude-encoded by qGAN as $\ket{\psi_\xi}=\sum_s \sqrt{p_s}\ket{b_s^\xi}$, it suffices to prepare $\hat{\xi}$ as a diagonal operator; then expectation-value evaluation yields weighted sums of the form $\sum_s p_s(\cdot)$.

\subsubsection{Binary-to-Pauli mapping and decision variable operator}
\label{subsubsec:uc_bin_to_pauli}

For the first-stage bit $x_i$ and the second-stage bit $b^{2\mathrm{nd}}_{i,j}$, we quantize them as
\begin{equation}
  x_i=\frac{1-\hat Z^{1\mathrm{st}}_i}{2},
  \qquad
  b^{2\mathrm{nd}}_{i,j}=\frac{1-\hat Z^{2\mathrm{nd}}_{i,j}}{2}
  \label{eq:bin_to_pauli_uc}
\end{equation}
We define the decision-variable operator corresponding to Eq.~\eqref{eq:y_gated_encoding} as
\begin{equation}
\begin{aligned}\label{eq:y_op_uc}
  &\hat y^{2\mathrm{nd}}_i
  :=
  \hat x^{1\mathrm{st}}_i
  \Bigg(P^{\mathrm{min}}_i +
  \sum_{j=1}^{n_{2\mathrm{nd},i}}
  \Delta^{2\mathrm{nd}}_i\,2^{j-1}\,\hat b^{2\mathrm{nd}}_{i,j}
  \Bigg)
  \\
  &\hat x^{1\mathrm{st}}_i:=\frac{1-\hat Z^{1\mathrm{st}}_i}{2},
  \quad
  \hat b^{2\mathrm{nd}}_{i,j}:=\frac{1-\hat Z^{2\mathrm{nd}}_{i,j}}{2} 
\end{aligned}
\end{equation}
With this definition, for any computational-basis state the bound $P^{\mathrm{min}}_i \hat x_i \le \hat y^{2\mathrm{nd}}_i \le P^{\mathrm{max}}_i \hat x_i$ holds.
Here, in the classical SAA/HUBO formulation, second-stage decision variables are introduced for each scenario as $\{b^{2\mathrm{nd}}_{i,s,j}\}$. In contrast, as explained in Sec.~\ref{subsubsec:2nd_stage_form}, the proposed method assigns an equivalent role through a superposition of second-stage-register quantum states conditioned on the scenario basis state $\ket{b_s^\xi}$.

\subsubsection{PV random-variable operator and problem Hamiltonian}
\label{subsubsec:uc_HP}

We define the random-variable operator on the scenario register as
\begin{equation}
  \hat \xi
  =
  \sum_{s=0}^{N-1}\xi_s\,\ket{b_s^\xi}\!\bra{b_s^\xi}
  \label{eq:xi_op_uc}
\end{equation}
We then define the diagonal Hamiltonian corresponding to Eq.~\eqref{eq:HUBO_obj_core} by
\begin{equation}
  H_P := H_P^{1\mathrm{st}} + H_P^{2\mathrm{nd}},
  \qquad
  H_P^{1\mathrm{st}} := \sum_{i=1}^{M} d_i\,\hat x^{1\mathrm{st}}_i
  \label{eq:HP_uc_total}
\end{equation}
\begin{equation}
  H_P^{2\mathrm{nd}}
  :=
  \sum_{i=1}^{M} c_i\,\hat y^{2\mathrm{nd}}_i
  + \lambda\Bigl(D\hat I - \hat\xi - \sum_{i=1}^{M}\hat y^{2\mathrm{nd}}_i\Bigr)^2
  \label{eq:HP_uc_2nd}
\end{equation}
where $\hat I$ denotes the identity operator.
Since $\hat y_i$ in Eq.~\eqref{eq:y_op_uc} is a polynomial in $\hat x_i$ and $\hat b^{2\mathrm{nd}}_{i,j}$, and $\hat x_i=(1-\hat Z^{1\mathrm{st}}_i)/2$ and $\hat b^{2\mathrm{nd}}_{i,j}=(1-\hat Z^{2\mathrm{nd}}_{i,j})/2$, $H_P$ can be expressed 
as a linear combination of Pauli-$Z$ tensor-product terms (Pauli-$Z$ strings).

\section{Results}
\label{sec:result}
In this section, we present numerical results for the proposed qGAN-QAOA method and validate (1) the validity of quantum state preparation for the scenario distribution (qGAN training), (2) the stability of the first-stage solutions obtained by qGAN-QAOA (dependence on initialization), and (3) the effectiveness of the method as a two-stage model from the viewpoint of the expected cost when the obtained solution is fixed.
First, in Section~\ref{subsec:exp_params}, we summarize the parameter settings of the stochastic UCP.
Next, in Section~\ref{subsec:qGAN-acc}, we train qGAN to learn the PV-output distribution and evaluate the agreement between the generated distribution and the discretized real-data distribution using the Jensen--Shannon divergence.
Finally, in Section~\ref{subsec:qganqaoa_results}, with the trained qGAN fixed, we optimize qGAN-QAOA and compare the marginal distribution of first-stage measurement outcomes (interpreted as first-stage solutions) and the expected cost computed using representative solutions with classical baselines (RP and EEV).

\subsection{Parameter Settings}
\label{subsec:exp_params}

In this subsection, we describe in detail the parameter settings used in our numerical experiments for the two-stage stochastic UCP defined in the previous section.

We fix the number of thermal generators to $M=3$ and set the demand to $D=2500\,\mathrm{kWh}$.
The uncertainty in PV output is modeled by a capacity factor (CF) that follows a Beta distribution,\footnote{For example, since the coefficient obtained by normalizing the generation amount by the maximum output has support on $[0,1]$, the Beta distribution is widely used.}
which is commonly employed in PV generation modeling.
Specifically, for each sample $s$, we generate the PV output $\xi_s$ as
\begin{equation}\label{eq:pv_from_cf}
\begin{aligned}
\mathrm{CF}_s &\sim \mathrm{Beta}(\alpha_{\mathrm{PV}}=3,\beta_{\mathrm{PV}}=7)\\
\xi_s &= \xi_{\max}\cdot\mathrm{CF}_s
\end{aligned}
\end{equation}
Under this setting, the mean of the Beta distribution is
\begin{equation}
\mathbb{E}[\mathrm{CF}] = \frac{\alpha_{\mathrm{PV}}}{\alpha_{\mathrm{PV}}+\beta_{\mathrm{PV}}}
= \frac{3}{3+7}=0.3
\label{eq:cf_mean}
\end{equation}
and therefore, by setting the maximum PV output to $\xi_{\max}=2500\,\mathrm{kWh}$, the mean PV output becomes
\begin{equation}
\mathbb{E}[\xi] = \xi_{\max}\cdot\mathbb{E}[\mathrm{CF}] = 2500\times 0.3 = 750\,\mathrm{kWh}
\label{eq:pv_mean}
\end{equation}
Accordingly, we assume that PV output is realized according to Eq.~\eqref{eq:pv_from_cf} over the interval $[\xi_{\min},\xi_{\max}] = [0, 2500]$.

\begin{table}[t]
\centering
\caption{Thermal power generator parameter settings ($M=3$).}
\label{tab:param_thermal_results}
\begin{tabular}{l c c c}
\hline
\hline
 Generator Parameters & Unit 1 & Unit 2 & Unit 3 \\
\hline
$P_i^{\min}$ (kWh) & 300 & 500 & 100 \\
$P_i^{\max}$ (kWh) & 750 & 1000 & 200 \\
$d_i$ (start-up cost: JPY) & 4000 & 5000 & 1000 \\
$c_i$ (generating cost: JPY/kWh) & 15 & 20 & 10 \\
\hline
\hline
\end{tabular}
\end{table}

For ease of implementation in a quantum circuit, the output of each thermal generator $i\in\mathcal{I}=\{1,2,3\}$ is restricted to two output levels.
That is, when generator $i$ is committed, its output can take either $P_i^{\min}$ or $P_i^{\max}$.
The specific generator parameters ($P_i^{\min},P_i^{\max},d_i,c_i$) are listed in Table~\ref{tab:param_thermal_results}.

The imbalance cost $\lambda$ is set to $18$ equally spaced values over the range $\lambda\in[30,200]$.

\subsection{Learning the PV uncertainty distribution using qGAN}
\label{subsec:qGAN-acc}

In this subsection, we train qGAN to learn the uncertainty distribution that the PV generation follows, and evaluate the training accuracy in terms of the agreement between the generated distribution and the real-data distribution.
For the quantum generator, we employ a TwoLocal circuit, and set the repetition depth to $p_{\mathrm{qGAN}}=n_{\xi}$ so as to increase the expressive power according to the number of qubits $n_{\xi}$ in the scenario register (the proposed framework can also be extended to other parameterized quantum circuits within $O(\mathrm{poly}(n_{\xi}))$.
).
As the discriminator, we use a fully connected neural network (MLP).
Specifically, we take as input the probability vector corresponding to the scenario distribution,
$\mathbf{p}=(p_0,\dots,p_{N-1})\in\mathbb{R}^N$,
and adopt an architecture with multiple hidden layers.

For the training data, we generate $N_{\mathrm{data}}=2000$ continuous samples based on Eq.~\eqref{eq:pv_from_cf} and regard them as real data for PV output.
Next, following Eq.~\eqref{eq:disc_rv}, we uniformly discretize the PV-output interval $[\xi_{\min},\xi_{\max}]=[0,2500]$ into
$N=4,8,16,32$ points (with $N=2^{n_{\xi}}$),
and construct a discrete empirical distribution (histogram) by binning the samples into the corresponding intervals.
Here, by changing the random seed, we repeat the generation of real data (continuous samples of size $N_{\mathrm{data}}=2000$) 15 times, and obtain 15 discrete distribution datasets by binning the resulting real data under the same discretization rule.
In what follows, we use 10 of them for qGAN training (train) and the remaining 5 for qGAN evaluation (test), and assess the generalization performance via the agreement on distribution datasets that are not used in training.

In qGAN training, at each epoch we estimate the scenario-generation probability distribution produced by the generator using samples obtained with $10{,}000$ measurement shots, and alternately optimize the discriminator and the generator.
We use Adam for optimization, and set the learning rates to $\eta_g=\eta_d=0.002$ for the generator and discriminator, respectively.
Training is terminated at 400 epochs, and we adopt the parameters at the epoch that attains the maximum agreement on the test data as the trained parameters $\boldsymbol{\theta}^{\ast}$.

As the evaluation metric, we use the agreement score $1-\mathrm{JS}$ based on the Jensen--Shannon divergence (JS divergence).
For probability distributions $P$ and $Q$, it is defined as
\begin{equation}
\mathrm{JS}(P,Q)
=\frac{1}{2}\mathrm{KL}\!\left(P\middle\|M\right)
+\frac{1}{2}\mathrm{KL}\!\left(Q\middle\|M\right),
\quad
M=\frac{1}{2}(P+Q)
\end{equation}
where $\mathrm{KL}(\cdot\|\cdot)$ denotes the Kullback--Leibler divergence.
We then compute the agreement score $1-\mathrm{JS}$ by normalizing $\mathrm{JS}$ so that $\mathrm{JS}\in[0,1]$ holds.

\begin{table}[t]
\centering
\small
\renewcommand{\arraystretch}{1.25}
\setlength{\tabcolsep}{6pt}
\caption{qGAN performance for different scenario resolutions.
We set $p_{\mathrm{qGAN}}=n_\xi=\log_2 N$ in the TwoLocal generator.
For each $N$, we trained qGAN with five random seeds.
We report $1-\mathrm{JS}$ as mean with standard deviation in parentheses
(e.g., $0.99983(3)=0.99983\pm0.00003$).
In subsequent experiments, we use the best-matching model (highest $1-\mathrm{JS}$); hence, the best score is also reported.}
\label{tab:qgan_js_results}
\begin{tabular}{c c c}
\hline\hline
$N$ & \makecell{$1-\mathrm{JS}$ (mean (std.))} & \makecell{$1-\mathrm{JS}$ (best of 5)} \\
\hline
4  & $0.99983(3)$  & 0.99986 \\
8  & $0.99917(28)$ & 0.99942 \\
16 & $0.99807(49)$ & 0.99842 \\
32 & $0.99423(68)$ & 0.99587 \\
\hline\hline
\end{tabular}
\end{table}

Table~\ref{tab:qgan_js_results} summarizes the qGAN performance for each discretization (scenario) size $N=4,8,16,32$.
For each $N$, we train qGAN with five different random seeds.
In this study, we set $p_{\mathrm{qGAN}}=n_{\xi}$, and thus the repetition depth of the generator increases as $N$ increases.
To assess reproducibility, we report the mean and standard deviation of the agreement score $1-\mathrm{JS}$ over the five seeds.
As shown in the table and Fig.~\ref{fig:accuracy_qGAN}, the mean agreement slightly decreases and the standard deviation tends to increase as $N$ increases, suggesting that learning a higher-resolution discrete distribution becomes relatively more challenging.
Nevertheless, a high agreement of $1-\mathrm{JS}\gtrsim 0.993$ is maintained for all settings.
In addition, since we use the trained generator with the best agreement (the highest $1-\mathrm{JS}$) for scenario-state preparation in the subsequent qGAN-QAOA experiments, we also list the best value (best of 5) for reference.
Overall, these results confirm that the proposed method can prepare the scenario-distribution quantum state with sufficient accuracy.

\begin{figure*}[t]
  \centering
  \captionsetup[subfigure]{labelformat=empty}


  \begin{subfigure}[t]{.49\linewidth}
    \centering
    \begin{overpic}[width=\linewidth,percent]{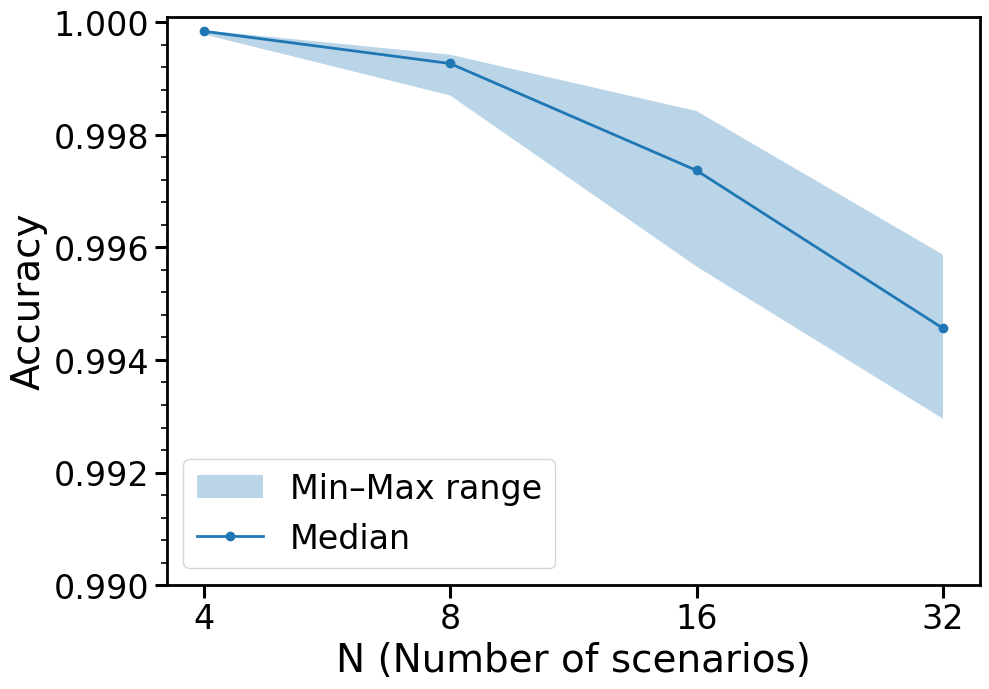}
    \put(2,75){\Large\bfseries (a)}
    \end{overpic}
    \subcaption{}\label{fig:accuracy_qGAN}
  \end{subfigure}\hfill
  \begin{subfigure}[t]{.49\linewidth}
    \centering
    \begin{overpic}[width=\linewidth,percent]{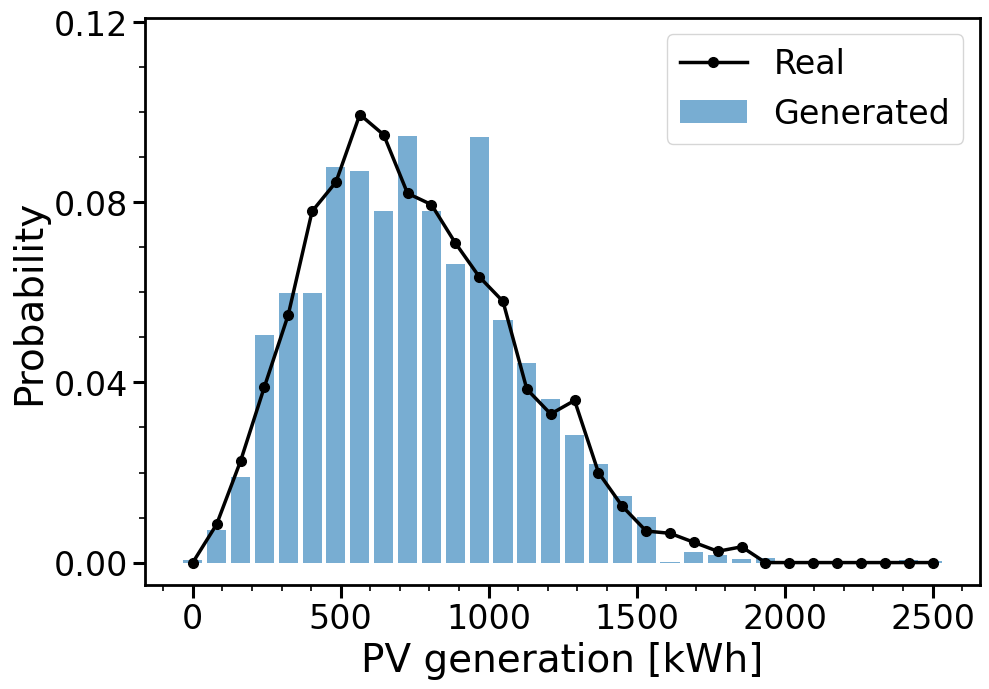}
    \put(2,75){\Large\bfseries (b)}
    \end{overpic}
    \subcaption{}\label{fig:trained_qGAN}
  \end{subfigure}

  \caption{Performance evaluation of the qGAN.
(a) For each discretization (scenario) size $N=4,8,16,32$, we evaluate the agreement score based on the Jensen--Shannon (JS) divergence over five random seeds used to generate the real data; the median and the range from the minimum to the maximum are shown.
(b) Comparison of the PV scenario probability distributions for $N=32$ between the real data (black line) and the generated data (blue line). We report the result with the highest agreement among the five qGAN training seeds.
  }
  \label{fig:qgan_performance}
\end{figure*}

Furthermore, Fig.~\ref{fig:trained_qGAN} shows, as a representative example for $N=32$, the real-data distribution (black line) and the generated distribution after training (bar plot).
The overall distributional shape agrees well, visually confirming that qGAN appropriately learns the uncertainty distribution of PV generation.

\subsection{Performance evaluation of qGAN-QAOA}
\label{subsec:qganqaoa_results}

\begin{figure*}[t]
  \centering
  \captionsetup[subfigure]{labelformat=empty}


  \begin{subfigure}[t]{.49\linewidth}
    \centering
    \begin{overpic}[width=\linewidth,percent]{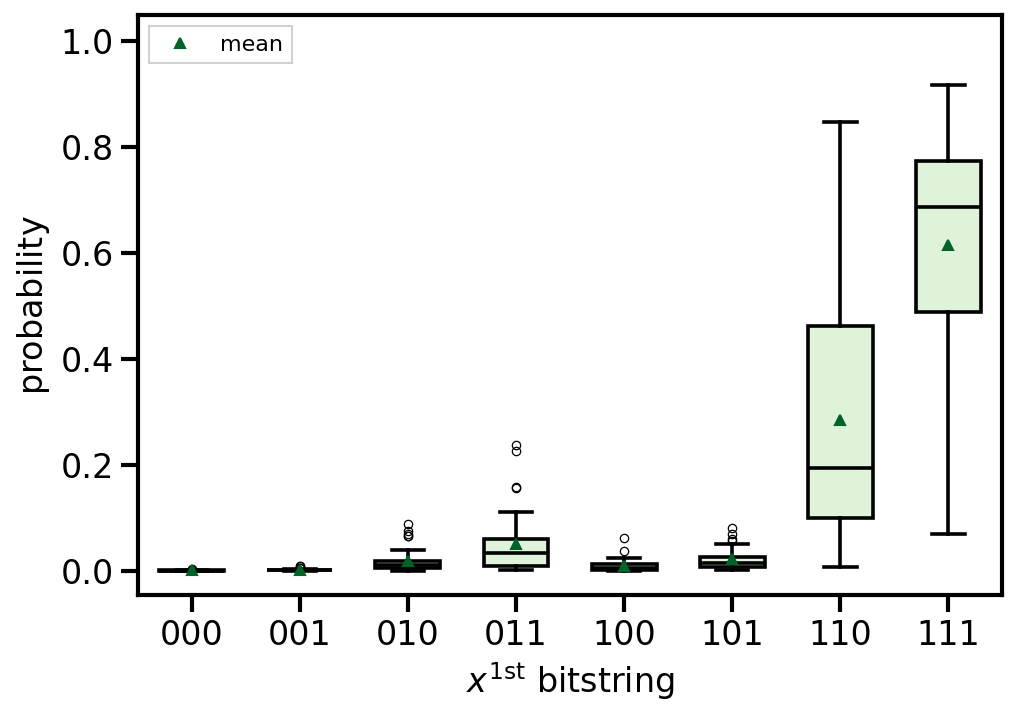}
    \put(2,75){\Large\bfseries (a)}
    \end{overpic}
    \subcaption{}\label{fig:qaoa_solution_dist}
  \end{subfigure}\hfill
  \begin{subfigure}[t]{.49\linewidth}
    \centering
    \begin{overpic}[width=\linewidth,percent]{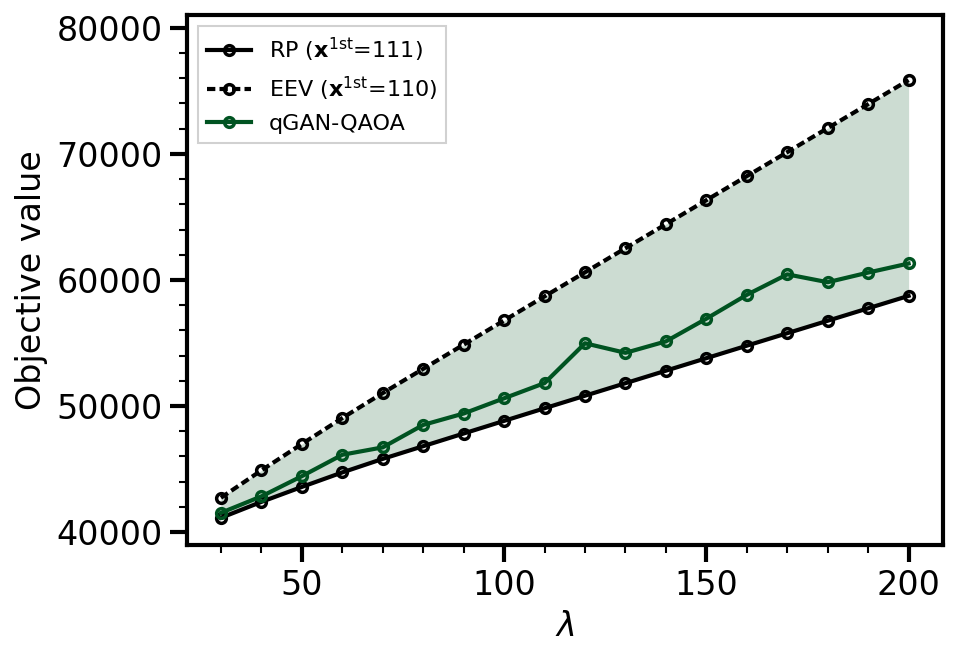}
    \put(2,75){\Large\bfseries (b)}
    \end{overpic}
    \subcaption{}\label{fig:objective_comparison}
  \end{subfigure}
    \caption{
Evaluation of the optimization results obtained by qGAN-QAOA.
(a) Marginal probability distribution of the first-stage measurement outcomes (interpreted as first-stage solutions) $\boldsymbol{x}_{\mathrm{1st}}$ obtained over 40 random seeds ($\lambda=30$).
The horizontal axis denotes the bit string representing $\boldsymbol{x}_{\mathrm{1st}}$, where the bits from left to right indicate the on/off status of generators No.~1, 2, and 3, respectively (on$=1$/off$=0$), and the vertical axis denotes the measurement probability for each seed.
The boxplot summarizes the distribution of the measurement probabilities (over 40 seeds) for each bit string: the box spans the first and third quartiles $Q_1$ and $Q_3$ (with $\mathrm{IQR}=Q_3-Q_1$), and the center line indicates the median.
The whiskers extend to the minimum and maximum data points within the range $[Q_1-1.5\,\mathrm{IQR},\,Q_3+1.5\,\mathrm{IQR}]$, and points outside this range are shown as outliers (circles).
The triangles indicate the mean measurement probability for each bit string.
(b) Comparison of the expected cost as a function of the imbalance cost $\lambda$.
The stochastic-programming baseline RP (black solid line) and EEV (black dashed line) are the evaluation values based on the $L_1$ formulation in Eq.~\eqref{eq:uc_sp_obj_xi} and Eq.~\eqref{eq:eev}, respectively.
For qGAN-QAOA (green), for each $\lambda$ we fix the first-stage solutions obtained from 40 random seeds as representative solutions and compute the expected cost using the $L_1$ two-stage evaluation in Eq.~\eqref{eq:Exp_cost_qGAN-QAOA}; the mean (green solid line) and the range (minimum--maximum, shaded band) are shown.
In the qGAN-QAOA optimization stage, for quantum-implementation reasons, we use the surrogate objective in Eq.~\eqref{eq:HUBO_obj_core}, where $\lvert \sigma \rvert$ is approximated by a quadratic penalty.
    }
    \label{fig:qgan_qaoa_performance}
\end{figure*}

In this subsection, we evaluate the stability (dependence on initialization) of the first-stage solutions $\boldsymbol{x}_{\mathrm{1st}}$ obtained by qGAN-QAOA and the expected cost when the obtained solutions are fixed.

In this numerical experiment, we consider the case where the PV uncertainty is uniformly discretized into $N=32$ scenarios.
That is, we use the qGAN generator circuit that produces the scenario distribution in Fig.~\ref{fig:trained_qGAN}, learned by qGAN with $n_{\xi}=\log_2 N=5$ qubits.
Moreover, since the number of generators is $M=3$, the numbers of qubits in the first-stage register and the second-stage register are
$n_{1\mathrm{st}}=n_{2\mathrm{nd}}=3$.
Therefore, the total number of qubits is
$n=n_{\xi}+n_{1\mathrm{st}}+n_{2\mathrm{nd}}=11$.

In the qGAN-QAOA optimization, to assess dependence on the initial QAOA parameters, for each $\lambda$ we generate $N_{\mathrm{seed}}=40$ random initializations of the QAOA parameters using different random seeds and run the optimization independently from each initialization.
We use Qiskit (version 1.1.1) for implementation and \texttt{AerSimulator} as the simulator.
The number of shots for expectation-value estimation is set to $\texttt{shots}=5\times10^{4}$.
As the classical optimizer, we employ \texttt{COBYLA} with the maximum number of iterations $\texttt{maxiter}=400$, the convergence tolerance $\texttt{tol}=10^{-3}$, and the initial trust-region radius $\texttt{rhobeg}=0.6$.
Among the estimated expectation values evaluated during the optimization, we adopt the parameters that yield the minimum value as the optimal parameters for the corresponding initialization.
In addition, we set the QAOA depths to $p_{1}=p_{2}=4$ for the first-stage and second-stage parts, respectively.

For each random seed, we measure the optimized circuit and obtain the marginal distribution of the first-stage measurement outcomes $\Pr(\boldsymbol{x}_{\mathrm{1st}})$.
In this study, we adopt the most probable outcome (MAP estimate),
\begin{equation}
\boldsymbol{x}_{\mathrm{1st}}^{\mathrm{opt}}
:= \arg\max_{\boldsymbol{x}_{\mathrm{1st}}\in\{0,1\}^{M}} \Pr(\boldsymbol{x}_{\mathrm{1st}})
\end{equation}
as a representative solution (if there are multiple ties for the maximum, we select one of them).

Fig.~\ref{fig:qaoa_solution_dist} shows, at $\lambda=30$, the marginal distribution of first-stage measurement outcomes (interpreted as first-stage solutions) $\boldsymbol{x}_{\mathrm{1st}}$ obtained over 40 random seeds.
The horizontal axis denotes the bit string representing $\boldsymbol{x}_{\mathrm{1st}}$, where the bits from left to right indicate the on/off status of generators No.~1, 2, and 3, respectively (on$=1$/off$=0$), and the vertical axis denotes the measurement probability for each seed.
The resulting solutions are mainly concentrated on two candidates, $\boldsymbol{x}_{\mathrm{1st}}=110$ and $\boldsymbol{x}_{\mathrm{1st}}=111$, while the other bit strings appear with low probability.
That is, although the local optimum to which the optimization converges can change depending on the initial parameter values, the first-stage solutions are strongly aggregated into a small set of candidates.

In particular, since the expected PV generation is $\mathbb{E}[\xi]=750\,\mathrm{kWh}$ from Eq.~\eqref{eq:pv_mean}, the expected-value problem (EV) naturally yields $\boldsymbol{x}_{\mathrm{1st}}=110$ as a plausible candidate when we consider the net demand $D-\mathbb{E}[\xi]=1{,}750\,\mathrm{kWh}$ as a reference.
On the other hand, taking into account that the scenario distribution learned by qGAN (Fig.~\ref{fig:trained_qGAN}) has a right tail, a more conservative solution $\boldsymbol{x}_{\mathrm{1st}}=111$ can be selected to suppress violations (penalties) of the supply--demand constraint in scenarios where PV generation falls below its expectation.
The concentration on these two candidates in Fig.~\ref{fig:qaoa_solution_dist} suggests that the proposed method narrows down to a few promising solution candidates under uncertainty.

Next, we evaluate the effectiveness of the solutions $\boldsymbol{x}_{\mathrm{1st}}^{\mathrm{opt}}$ obtained by the qGAN-QAOA method in terms of the expected cost.
Here, the quantity referred to as the expected cost is the two-stage evaluation value based on the original $L_1$ objective in Eq.~\eqref{eq:uc_sp_obj_xi}, and it should be distinguished from the expected value of the surrogate objective in Eq.~\eqref{eq:HUBO_obj_core} that is minimized in the quantum circuit.
\begin{equation}\label{eq:Exp_cost_qGAN-QAOA}
\begin{aligned}
\mathrm{C}(\boldsymbol{x}_{\mathrm{1st}}^{\mathrm{opt}})
&:=
\sum_{i=1}^{M} d_i x_{\mathrm{1st},i}^{\mathrm{opt}}
+\sum_{s=0}^{N_{\mathrm{test}}-1} \tilde{p}_s
\min_{\boldsymbol{y}_s \in \mathcal{Y}(\boldsymbol{x}_{\mathrm{1st}}^{\mathrm{opt}})}
\Bigg{\{}
\sum_{i=1}^{M} c_i y_{i,s}
+\lambda |\sigma_s |
\Bigg{\}},\\
\mathcal{Y}(\boldsymbol{x}_{\mathrm{1st}}^{\mathrm{opt}})
&:=
\big{\{}
\boldsymbol{y}_s \in \mathbb{R}_{\ge 0}^{M} \big| \, 
\sigma_s = D-\tilde{\xi}_s-\sum_{i=1}^{M} y_{i,s},
\\
&\qquad\qquad\qquad\;\;\, y_{i,s} \in \left\{x_{\mathrm{1st},i}^{\mathrm{opt}} P_i^{\min},\ x_{\mathrm{1st},i}^{\mathrm{opt}} P_i^{\max}\right\}
\; (\forall i=1,\dots,M)
\big{\}}
\end{aligned}
\end{equation}
In this study, for qGAN training and the qGAN-QAOA optimization, we use the uniformly discretized scenario distribution with $N=32$ scenarios, $\{(\xi_s,p_s)\}_{s=0}^{N-1}$ (the distribution for qGAN training and optimization).
In contrast, to evaluate the performance of the first-stage solutions obtained by the proposed method (i.e., to compute the expected cost), we replace it---independently of the discretized distribution used for training---with a representative test-scenario set $\{(\tilde{\xi}_s,\tilde{p}_s)\}_{s=0}^{N_{\mathrm{test}}-1}$ of size $N_{\mathrm{test}}=200$, which is extracted from $N_{\mathrm{data}}$ real-data samples via the quantile method, and compute the expectation in Eq.~\eqref{eq:Exp_cost_qGAN-QAOA}.
Here, we assign equal probabilities $\tilde{p}_s=1/N_{\mathrm{test}}$.
By separating the scenario sets used for training/optimization and for evaluation in this manner, we can examine whether the solutions obtained by the proposed method are not overly fitted to the training distribution (generalization ability).

As classical baselines, we use the expected cost obtained by solving the stochastic program (RP: Recourse Problem) and the expected cost obtained by fixing the solution from the expected-value problem (EV) and evaluating it (EEV: Expected Value of the Expected Value solution).
RP is the objective value obtained by solving the stochastic program, i.e., the optimal value of Eq.~\eqref{eq:uc_sp_obj_xi}.
On the other hand, the expected-value problem (EV) is defined using the expected PV generation
$\mathbb{E}[\xi] \approx \frac{1}{N_{\mathrm{test}}}\sum_{s=0}^{N_{\mathrm{test}}-1} \tilde{\xi}_s$ as
\begin{equation}\label{eq:EV_problem}
    \begin{aligned}
    (\boldsymbol{x}_{\mathrm{1st}}^{\mathrm{EV}},\, \boldsymbol{y}^{\mathrm{EV}})
    \in
    &\arg
    \min_{\substack{\boldsymbol{x}_{\mathrm{1st}}\in\{0,1\}^{M}\\ \boldsymbol{y}\in\mathbb{R}_{\ge 0}^{M}}}
    \Bigg{\{}
    \sum_{i=1}^{M} d_i x_{{\mathrm{1st}},i}
    +\sum_{i=1}^{M} c_i y_i
    +\lambda |\sigma|
    \Bigg{\}} \\
    &
    \qquad
    \mathrm{s.t.} \;
    \sigma=D-\mathbb{E}[\xi]-\sum_{i=1}^{M} y_i\\
    &
    \quad\quad\quad\;\; y_{i} \in \left\{x_{\mathrm{1st},i} P_i^{\min},\ x_{\mathrm{1st},i} P_i^{\max}\right\},
    \quad \forall i\\
    \end{aligned}
\end{equation}
Furthermore, EEV is defined by fixing the first-stage solution $\boldsymbol{x}_{\mathrm{1st}}^{\mathrm{EV}}$ obtained by EV, evaluating the cost in each scenario, and then taking the expectation:
\begin{equation}\label{eq:eev}
\begin{aligned}
\mathrm{EEV}
&:=
\sum_{i=1}^{M} d_i x_{\mathrm{1st},i}^{\mathrm{EV}}
+\sum_{s=0}^{N_{\mathrm{test}}-1} 
\tilde{p}_s \;
\min_{\boldsymbol{y}_s \in \mathcal{Y}(x_{\mathrm{1st}}^{\mathrm{EV}})}
\Bigg\{
\sum_{i=1}^{M} c_i y_{i,s}
+\lambda |\sigma_s|
\Bigg\}\\
&\;
\;\qquad \qquad \mathrm{s.t.} \;
\sigma_s=D-\tilde{\xi}_s-\sum_{i=1}^{M} y_{i,s},\ \forall s
\end{aligned}
\end{equation}
Here, RP/EEV are also computed using the evaluation scenario set {$(\tilde\xi_s,\tilde p_s)$}.

Fig.~\ref{fig:objective_comparison} compares the expected cost as a function of the imbalance cost $\lambda$.
In addition to RP (black solid line) and EEV (black dashed line), for qGAN-QAOA (green) we evaluate the expected cost by fixing $\boldsymbol{x}^{\mathrm{1st}}$ obtained from 40 random seeds for each $\lambda$, and report the mean (green solid line) together with the range (minimum--maximum, shaded band).
As shown in Fig.~\ref{fig:objective_comparison}, as $\lambda$ increases, EEV tends to deviate from RP, whereas qGAN-QAOA tends to yield an expected cost closer to RP.
The width of the shaded band is attributed to dependence on the initial parameter values, yet the mean value consistently remains close to RP over a wide range of $\lambda$,
indicating that the proposed method is effective as a two-stage stochastic-programming decision-making approach that accounts for uncertainty.

The above results indicate that (i) qGAN can approximate the PV distribution in this case with high agreement, (ii) qGAN-QAOA outputs a probability distribution that concentrates the first-stage solutions into a small set of promising candidates, and (iii) the expected cost of the representative solution tends to be closer to RP than to EEV.
That is, the proposed method provides a consistent implementation candidate for executing two-stage decision-making on a single variational quantum circuit, where uncertainty is treated as a distribution rather than a point estimate (expectation).
On the other hand, note that the above comparison is performed on a simulator, and that the expected-cost evaluation uses an evaluation scenario set that is separate from the training scenario set.

\section{Discussion}\label{sec:discussion} 

In this section, we discuss the scaling of the quantum-circuit size of the proposed qGAN-QAOA method, namely, the scaling of the gate count and circuit depth.
The goal of the proposed method is to reduce the dependence on the scenario count $N$ that appears in classical two-stage stochastic programming based on SAA when computing expectations.
In particular, in Section~\ref{subsec:WHT_Pauli}, under uniform discretization of a continuous uncertainty, we show that the random-variable operator $\hat{\xi}$ admits a Pauli-$Z$ expansion based on the Walsh--Hadamard transform, and that the number of Pauli-$Z$ terms arising from $\hat{\xi}$ can be bounded by $O(\mathrm{poly}(\log N))$, thereby mitigating the dependence on $N$.
In the subsequent Section~\ref{subsec:computational_complexity}, we evaluate the computational complexity when the discretization error and the expectation-estimation error are both set to $O(\varepsilon)$, and compare the dominant per-iteration complexity terms of classical two-stage stochastic programming based on SAA with those of the proposed method.

\subsection{Walsh--Hadamard based Pauli expansion and circuit scaling}
\label{subsec:WHT_Pauli}

In this subsection, we show that the random-variable operator in Eq.~\eqref{eq:RandomVariableOperator} can be represented as a linear combination of Pauli-$Z$ strings via the Walsh--Hadamard transform (WHT).
Note that representations and circuit synthesis (i.e., compiling the operator into a gate sequence) of diagonal operators based on the WHT have already been systematized in prior work~\cite{Welch2014DiagonalWalsh}.
More recently, studies have also been reported that interpret the computation of Pauli-expansion coefficients as the (fast) Walsh--Hadamard transform and compute them efficiently~\cite{Georges2025PauliFWHT}.
The main focus of this paper is not to propose a new synthesis method for general diagonal operators, but rather to clarify that, for the scenario register appearing in two-stage stochastic programming, the discretization structure---in particular, uniform discretization---makes the WHT coefficients sparse and, as a result, can suppress the dependence of circuit cost on the scenario count $N$.

In what follows, we show in particular that, when the realizations are given by an arithmetic progression, the number of required Pauli-$Z$ terms can be reduced to $n_\xi+1=O(\log N)$.
By comparison with an existing scenario-projector-based implementation~\cite{rotello2024_expected_value}, we then discuss how the dependence on the scenario count $N$ can be improved.

\subsubsection{Pauli strings and Walsh--Hadamard transform}

Let the computational basis of the scenario register consisting of $n_\xi$ qubits be
\begin{equation*}
  \{\ket{b^\xi_s}\}_{s=0}^{2^{n_\xi}-1},\quad
  b^\xi_s = b^\xi_{s,n_\xi-1}\dots b^\xi_{s,1} b^\xi_{s,0} \in \{0,1\}^{n_\xi}
\end{equation*}
The Pauli-$Z$ operator satisfies
\begin{equation*}
  Z\ket{0} = \ket{0},\quad
  Z\ket{1} = -\ket{1}
\end{equation*}
and thus, for a bit $b^\xi_{s,i} \in \{0,1\}$, we can write
\begin{equation*}
  Z^{b^\xi_{s,i}}\ket{b^\xi_{s,i}} = (-1)^{b^\xi_{s,i}}\ket{b^\xi_{s,i}}
\end{equation*}
For an $n_\xi$-bit string $j = j_{n_\xi-1}\dots j_1 j_0 \in \{0,1\}^{n_\xi}$, define a Pauli string composed of tensor products of $Z$ as
\begin{equation*}
  \hat{P}_j := Z^{j_{n_\xi-1}}\otimes \cdots \otimes Z^{j_1}\otimes Z^{j_0}
\end{equation*}
Then, the action of $\hat{P}_j$ can be written as
\begin{align*}
  \hat{P}_j \ket{b^\xi_s}
  &= \bigotimes_{i=0}^{n_\xi-1} Z^{j_i} \ket{b^\xi_{s,i}} \\
  &= \bigotimes_{i=0}^{n_\xi-1} (-1)^{j_i b^\xi_{s,i}} \ket{b^\xi_{s,i}} \\
  &= (-1)^{j\cdot s} \ket{b^\xi_s}
\end{align*}
where
\begin{equation*}
  j\cdot s := \sum_{i=0}^{n_\xi-1} j_i b^\xi_{s,i} \pmod 2
\end{equation*}
is the bitwise inner product over $\mathbb{Z}_2^{n_\xi}$.

Next, define the $2^{n_\xi} \times 2^{n_\xi}$ matrix $H_{n_\xi}$ by
\begin{equation*}
  (H_{n_\xi})_{j,s} := (-1)^{j\cdot s}, \quad j,s \in \{0,\dots,2^{n_\xi}-1\}
\end{equation*}
Then, $H_{n_\xi}$ is the Walsh--Hadamard matrix given by the tensor product of the single-qubit Hadamard transform $H_1$:
\begin{equation*}
  H_{n_\xi} = \bigotimes_{i=1}^{n_\xi} H_1, \quad
  H_1 = \begin{pmatrix} 1 & 1 \\ 1 & -1 \end{pmatrix}
\end{equation*}

\subsubsection{Pauli expansion of diagonal operators}

Consider the random-variable operator defined by a diagonal matrix of size $2^{n_\xi} \times 2^{n_\xi}$:
\[
  \hat{\xi} = \mathrm{diag}(\xi_0,\xi_1,\dots,\xi_{2^{n_\xi}-1})
\]
Define the column vector obtained by listing its diagonal entries from top to bottom as
\[
  \boldsymbol{\xi}
  :=
  \begin{pmatrix}
    \xi_0 \\ \xi_1 \\ \vdots \\ \xi_{2^{n_\xi}-1}
  \end{pmatrix}
\in \mathbb{R}^{2^{n_\xi}}
\]
Similarly, introduce the column vector collecting the expansion coefficients of the Pauli-$Z$ strings $\hat{P}_j$:
\[
  \boldsymbol{c}
  :=
  \begin{pmatrix}
    c_0 \\ c_1 \\ \vdots \\ c_{2^{n_\xi}-1}
  \end{pmatrix}
\in \mathbb{R}^{2^{n_\xi}}
\]
and write the Pauli-$Z$ expansion of $\hat{\xi}$ as
\begin{equation}
  \hat{\xi}
  = \sum_{j=0}^{2^{n_\xi}-1} c_j \hat{P}_j
  \label{eq:omega_Pauli_expansion}
\end{equation}
Using $\hat{\xi}\ket{b^\xi_s} = \xi_s \ket{b^\xi_s}$ and $\hat{P}_j\ket{b^\xi_s} = (-1)^{j\cdot s}\ket{b^\xi_s}$, we obtain
\[
  \xi_s
  = \sum_{j=0}^{2^{n_\xi}-1} c_j (-1)^{j\cdot s}
\]
In vector form, this is
\[
  \boldsymbol{\xi} = H_{n_\xi} \boldsymbol{c}
\]
and therefore,
\begin{equation}
  \boldsymbol{c} = \frac{1}{2^{n_\xi}} H_{n_\xi} \boldsymbol{\xi}
  \label{eq:WHT_coeff_general}
\end{equation}
That is, the Pauli-$Z$ expansion coefficients of the random-variable operator $\hat{\xi}$ are uniquely determined by the Walsh--Hadamard transform of the diagonal-entry vector.
(The coefficient computation itself can be performed in $O(N\log N)$ via the FWHT~\cite{Georges2025PauliFWHT}.)

\subsubsection{Scaling for number of scenarios}
\label{subsubsec:Scaling_scenarios}

As expressed in Eq.~\eqref{eq:disc_rv}, when the scenario realizations $\xi_s$ are given by an equally spaced arithmetic progression, many of the WHT coefficients of the random-variable operator become zero, and the number of nonzero coefficients can be reduced to $n_\xi+1 = O(\log N)$.
Note that this sparsity is a property of the random-variable operator $\hat{\xi}$ and its low-degree polynomials, and it does not claim that the entire problem Hamiltonian in a general two-stage stochastic program is always sparse to a comparable extent.

\begin{proposition}
  \label{prop:arithmetic_WHT}
  Let the number of scenarios be $N = 2^{n_\xi}$, and suppose that the realization $\xi_s$ of the random variable corresponding to scenario $s\in\{0,\dots,N-1\}$ is given by
  \begin{equation}\label{eq:arithmetic_seq}
    \xi_s = \xi_{\min} + s\,\Delta\xi,\quad
    \Delta\xi = \frac{\xi_{\max}-\xi_{\min}}{N-1}
  \end{equation}
  Let $\boldsymbol{\xi}$ be the vector determined by Eq.~\eqref{eq:arithmetic_seq},
  \[
    \boldsymbol{\xi}
    =
    \begin{pmatrix}
      \xi_0 \\ \xi_1 \\ \vdots \\ \xi_{N-1}
    \end{pmatrix}
    \in \mathbb{R}^{N}
  \]
  and consider the Walsh--Hadamard transform coefficients defined by Eq.~\eqref{eq:WHT_coeff_general},
  \[
    c_j
    = \frac{1}{2^{n_\xi}} \sum_{s=0}^{2^{n_\xi}-1} (-1)^{j\cdot s} \,\xi_s,
    \quad j=0,\dots,2^{n_\xi}-1
  \]
  Let the Hamming weight of the binary representation $j = \sum_{i=0}^{n_\xi-1} j_i 2^i$ be written as
  \[
    w(j) := \sum_{i=0}^{n_\xi-1} j_i
  \]
  Then, the coefficient $c_j$ satisfies
  \begin{equation}\label{eq:WHT_coeff_arith}
    c_j =
    \begin{cases}
      \displaystyle
      \xi_{\min}
      + \Delta\xi\,\frac{2^{n_\xi}-1}{2}
      &(w(j)=0)\\
      \displaystyle
      -\,\Delta\xi\,\frac{j}{2}
      &(w(j)=1)\\
      0
      &(w(j)\ge 2)
    \end{cases}
  \end{equation}
  In particular, when $w(j)=1$, we have $j=2^i$ ($i=0,1,\dots,n_\xi-1$), in which case $c_j = -\,\Delta\xi\,2^{i-1}$.
  Therefore, the nonzero coefficients are restricted to $j=0$ and $j=2^{i}$ ($i=0,1,\dots,n_\xi-1$), and their number is $n_\xi+1$.
\end{proposition}
\noindent A proof of this proposition is provided in Appendix~\ref{app:proof_arithmetic_WHT}.

From Proposition~\ref{prop:arithmetic_WHT}, when the realizations of the discrete random variable are given by an arithmetic progression, the random-variable operator can be represented using only $n_\xi+1$ terms, namely, the identity matrix $I$ and $n_\xi$ single-qubit Pauli-$Z$ operators.
Within the existing WHT/Walsh-series-based framework for diagonal-operator representations (e.g.,~\cite{Welch2014DiagonalWalsh}), this result explicitly shows that the discretization scheme considered in this study yields a particularly sparse spectrum.

In Rotello et al.'s scenario-projector-based approach~\cite{rotello2024_expected_value}, the problem Hamiltonian $H_{P}$ is given, using the projector $\ket{\xi_s}\!\bra{\xi_s}$ onto the computational basis state $\ket{\xi_s}$ corresponding to scenario $\xi_s$ as
\begin{equation}
  H_{P}
  = \sum_{s=0}^{N-1} H_{\mathrm{QAOA}}(\hat{Z},\xi_s) \otimes \ket{\xi_s}\!\bra{\xi_s}
  \label{eq:rotello_HP}
\end{equation}
Here, $H_{\mathrm{QAOA}}(\hat{Z},\xi_s)$ is a QAOA-type cost Hamiltonian corresponding to the second-stage problem under scenario $\xi_s$, and is expressed as a polynomial in the Pauli operators $\hat{Z}$ acting on the decision-variable registers.
In this formulation, the projectors require embedding a scenario-dependent Hamiltonian $H_{\mathrm{QAOA}}(\hat{Z},\xi_s)$ for each scenario $s$ into the circuit, so that the number of terms with respect to the scenario count $N$ grows at least as $O(N)$.
Moreover, the projector $\ket{\xi_s}\!\bra{\xi_s}$ is generally implemented as a multi-controlled operation on $n_\xi=\log N$ qubits, and its decomposition into two-qubit gates requires many controlled gates, leading to a substantial implementation burden through increased circuit depth and gate count.

In contrast, in the proposed WHT-based formulation, the random-variable operator can be represented as a linear combination of Pauli-$Z$ operators depending on the structure of the scenario realizations.
In particular, under the condition that the realizations form an arithmetic progression, Proposition~\ref{prop:arithmetic_WHT} shows that the random-variable operator $\hat{\xi}$ can be expressed using only $\log N+1$ Pauli-$Z$ terms (the identity matrix $I$ and $n_\xi=\log N$ single-qubit $Z$ operators).
Here, consider the case where the problem Hamiltonian $H_P$ is given, as in Eq.~\eqref{eq:HP_uc_2nd}, by polynomials in $\hat{\xi},\;\hat{\boldsymbol{x}}^{1\mathrm{st}},\;\hat{\boldsymbol{y}}^{2\mathrm{nd}}$ with degrees $\deg_\xi,\;\deg_x,\;\deg_y\leq2$.
Then, for the total number of qubits $n=n_{\xi}+n_{1\mathrm{st}}+n_{2\mathrm{nd}}$, we can express $H_P^{2\mathrm{nd}}$ as a linear combination of Pauli-$Z$ strings $\hat P_j$:
\begin{equation}
  H_P^{2\mathrm{nd}}
  =
  \sum_{j=1}^{L} w_j \hat P_j
  \;+\;\mathrm{const.},
  \quad
  \hat P_j \in \{\hat I,\hat Z\}^{\otimes n}
  \label{eq:HP2_pauliZ_linear_comb}
\end{equation}
where $L$ denotes the number of Pauli-$Z$ strings.
By exploiting the sparse representation based on uniform discretization and the WHT in this paper, the term count can be bounded by $O(\mathrm{poly}(n))$.
Furthermore, for a Pauli-$Z$ string of weight $m\leq n$,
$\hat Z_{\alpha_1}\hat Z_{\alpha_2}\cdots \hat Z_{\alpha_m}$, the corresponding exponential operator can be implemented using $O(m)$ one- and two-qubit gates by applying CNOT and $R_Z$ gates~\cite{fleury2025}:
\begin{equation}
\label{eq:pauliZ_string_synthesis}
\begin{aligned}
&\exp\!\bigl(-i t\, \hat Z_{\alpha_1}\hat Z_{\alpha_2}\cdots \hat Z_{\alpha_m}\bigr)\\
&=
\Biggl(\prod_{\ell=m}^{2} \mathrm{CX}_{\alpha_\ell,\alpha_{\ell-1}}\Biggr)\,
R_{Z,\alpha_1}(2t)\,
\Biggl(\prod_{\ell=2}^{m} \mathrm{CX}_{\alpha_\ell,\alpha_{\ell-1}}\Biggr)   
\end{aligned}
\end{equation}
Therefore, the problem-Hamiltonian gate $\exp(-i\gamma H_P^{2\mathrm{nd}})$ can be implemented as a product of commuting terms, and the gate count and circuit depth scale approximately as $\sum_{j=1}^{L} O(\mathrm{wt}(\hat P_j))$, i.e., proportionally to $O(\mathrm{poly}(n))$.
Focusing on the dependence on the scenario count $N$, whereas the projector-based construction increases as $O(N)$, in the proposed construction the number of Pauli-$Z$ terms arising from $\hat{\xi}$, which appear when replacing each power of $\hat{\xi}$ by its Pauli expansion, is bounded by $O(\deg_\xi\,\log N)$, yielding
\begin{equation}\label{eq:scaling_qGAN-QAOA}
    O(\deg_\xi\,\log N)=O(\mathrm{poly}(\log N))
\end{equation}
As the number and dimension of the decision variables increase, they can dominate the overall circuit size; however, as discussed in Section~\ref{subsec:computational_complexity}, with respect to the bottleneck induced by increasing the scenario count, the proposed method provides more favorable scaling behavior from the viewpoint of computational complexity.

\subsubsection{Quantum circuit scaling in stochastic UCP}
\label{subsubsec:scaling_stochastic_UCP}

\begin{figure*}[t]
  \centering
  \captionsetup[subfigure]{labelformat=empty}


  \begin{subfigure}[t]{.245\linewidth}
    \centering
    \begin{overpic}[width=\linewidth,percent]{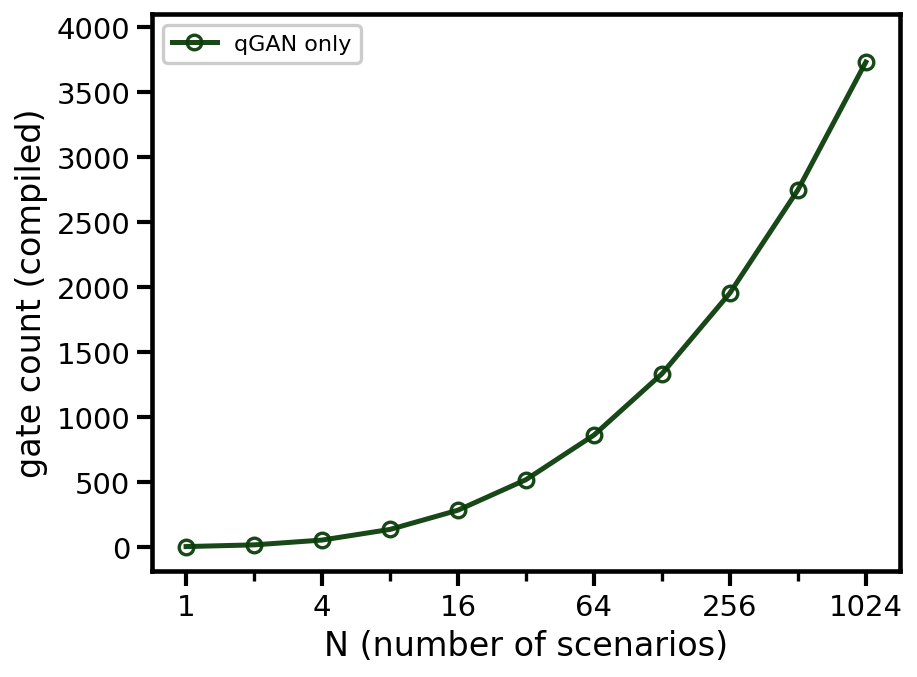}
      \put(2,80){\Large\bfseries (a)}
    \end{overpic}
    \subcaption{}\label{fig:qgan_only_gate}
  \end{subfigure}\hfill
  \begin{subfigure}[t]{.245\linewidth}
    \centering
    \begin{overpic}[width=\linewidth,percent]{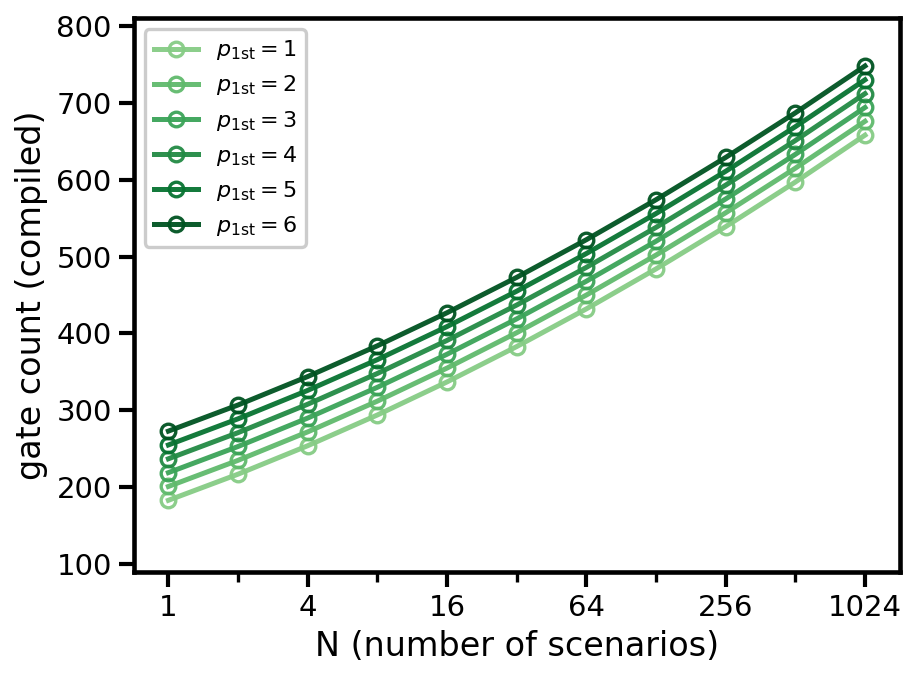}
      \put(2,80){\Large\bfseries (b)}
    \end{overpic}
    \subcaption{}\label{fig:no_qgan_p1_gate}
  \end{subfigure}\hfill
  \begin{subfigure}[t]{.245\linewidth}
    \centering
    \begin{overpic}[width=\linewidth,percent]{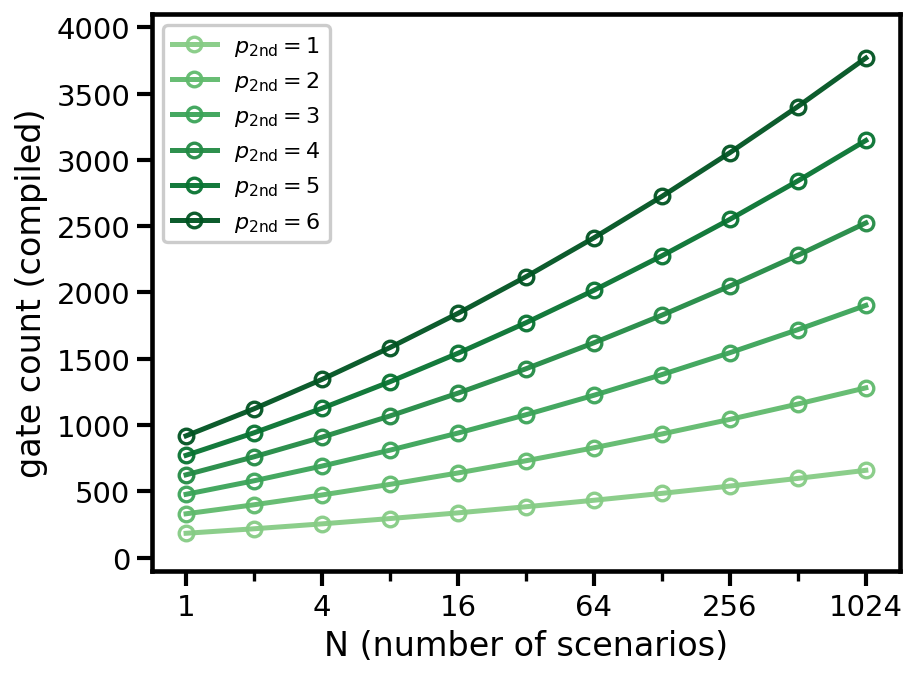}
      \put(2,80){\Large\bfseries (c)}
    \end{overpic}
    \subcaption{}\label{fig:no_qgan_p2_gate}
  \end{subfigure}\hfill
  \begin{subfigure}[t]{.245\linewidth}
    \centering
    \begin{overpic}[width=\linewidth,percent]{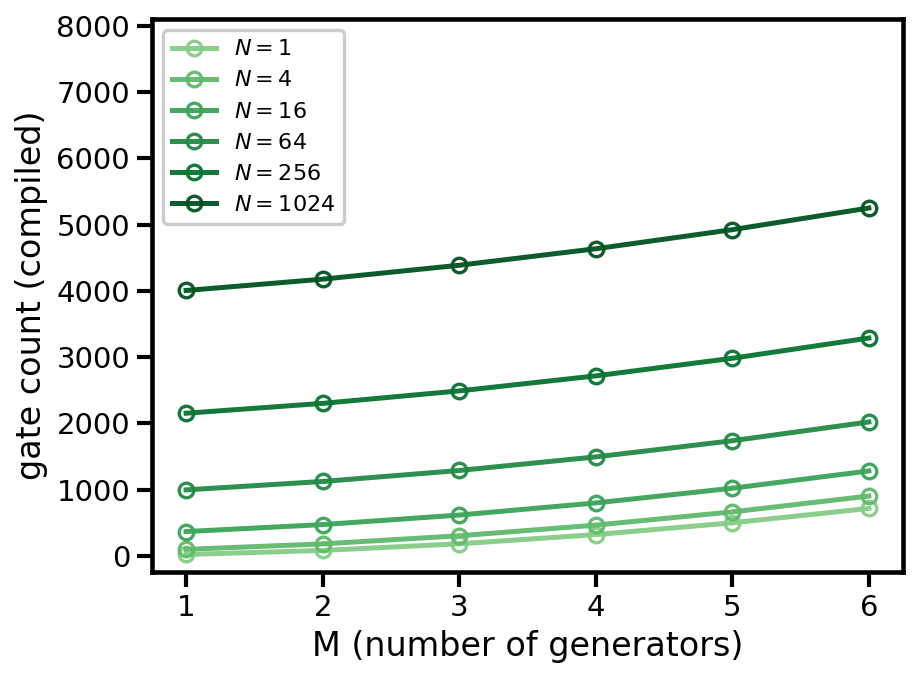}
      \put(2,80){\Large\bfseries (d)}
    \end{overpic}
    \subcaption{}\label{fig:m_scale_gate}
  \end{subfigure}

  \vspace{2mm} 

  \begin{subfigure}[t]{.245\linewidth}
    \centering
    \begin{overpic}[width=\linewidth,percent]{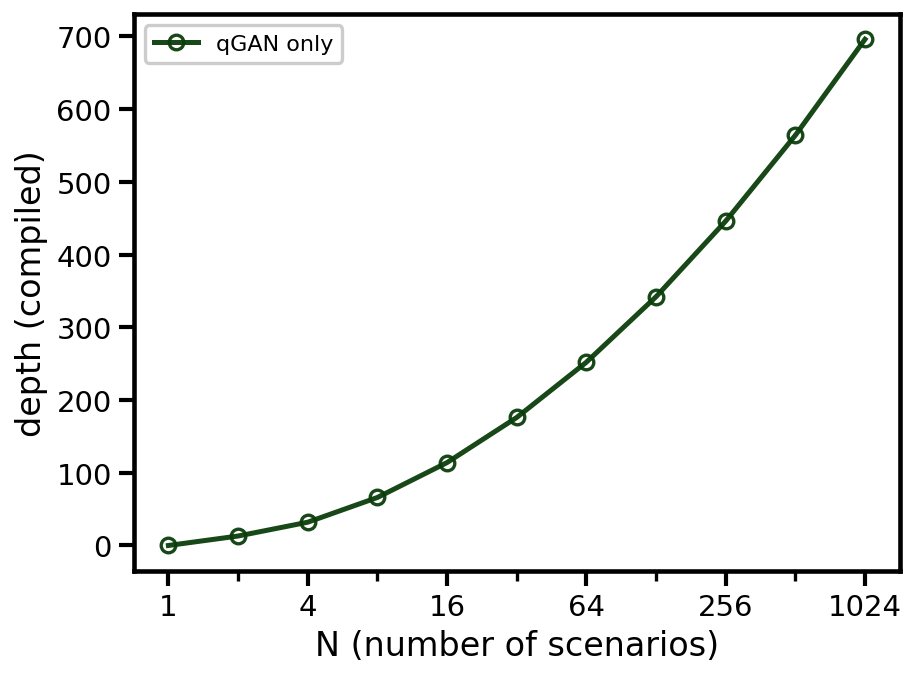}
      \put(2,80){\Large\bfseries (e)}
    \end{overpic}
    \subcaption{}\label{fig:qgan_only_depth}
  \end{subfigure}\hfill
  \begin{subfigure}[t]{.245\linewidth}
    \centering
    \begin{overpic}[width=\linewidth,percent]{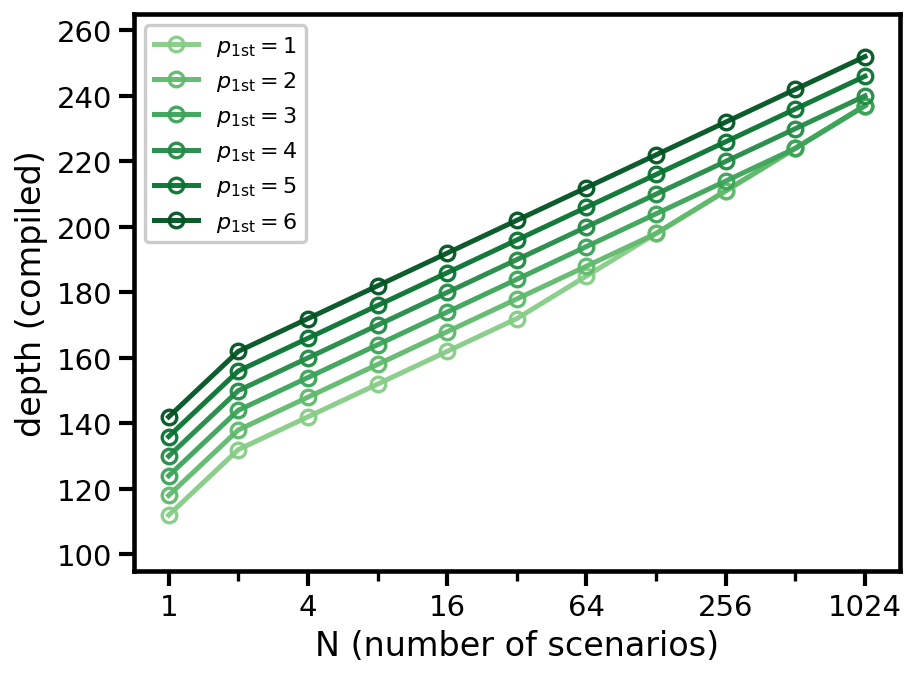}
      \put(2,80){\Large\bfseries (f)}
    \end{overpic}
    \subcaption{}\label{fig:no_qgan_p1_depth}
  \end{subfigure}\hfill
  \begin{subfigure}[t]{.245\linewidth}
    \centering
    \begin{overpic}[width=\linewidth,percent]{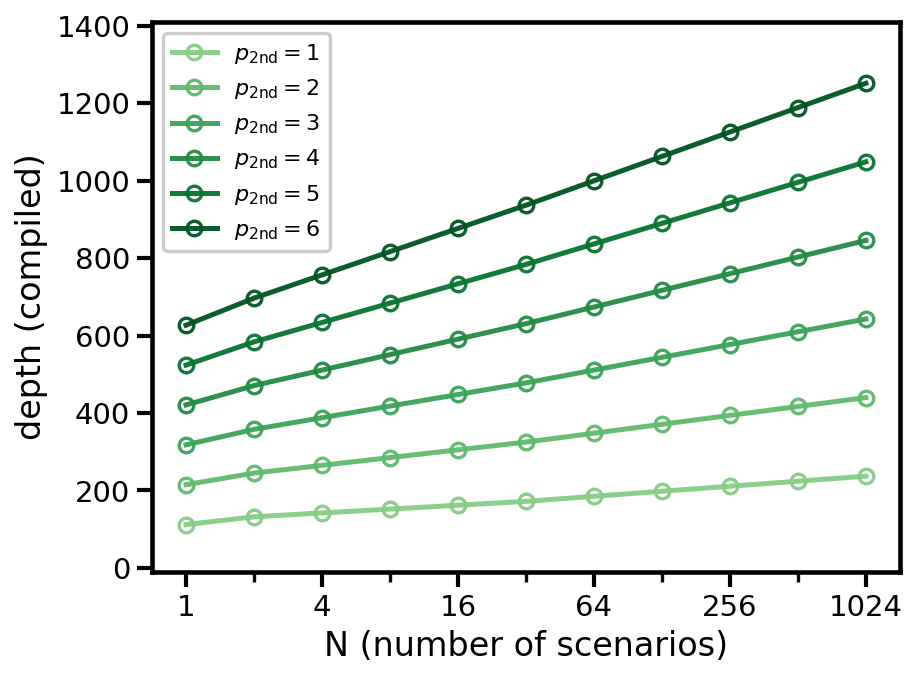}
      \put(2,80){\Large\bfseries (g)}
    \end{overpic}
    \subcaption{}\label{fig:no_qgan_p2_depth}
  \end{subfigure}\hfill
  \begin{subfigure}[t]{.245\linewidth}
    \centering
    \begin{overpic}[width=\linewidth,percent]{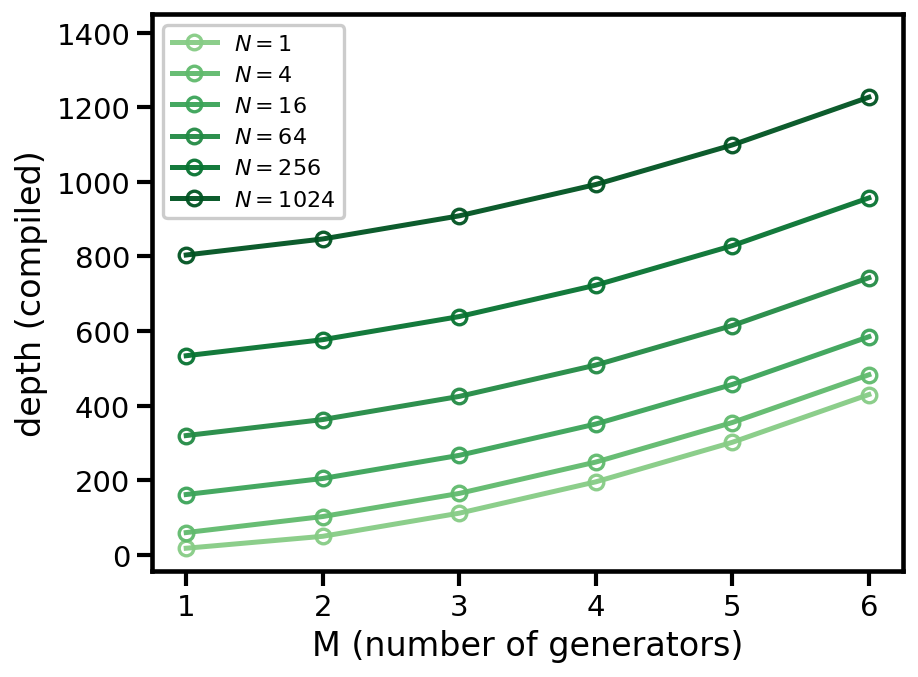}
      \put(2,80){\Large\bfseries (h)}
    \end{overpic}
    \subcaption{}\label{fig:m_scale_depth}
  \end{subfigure}
    \caption{
Scaling of the compiled gate count and circuit depth with respect to the uncertainty discretization, i.e., the scenario count $N$, and the number of generators $M$.
Panels (a)--(d) show the compiled gate count, and panels (e)--(h) show the compiled depth.
Panels (a) and (e) evaluate the size of qGAN (the scenario generator) alone as a function of $N$.
Panels (b) and (f) show the $N$-dependence of the circuit excluding the qGAN contribution in the proposed method, when varying $p_{\mathrm{1st}}$.
Panels (c) and (g) show, likewise, the $N$-dependence of the circuit excluding the qGAN contribution, when varying $p_{\mathrm{2nd}}$.
Panels (d) and (h) compare the circuit size of the entire proposed method including qGAN ($p_{\mathrm{1st}}=p_{\mathrm{2nd}}=1$) across multiple values of $N$, when varying $M$.
    }
\label{fig:scaling_stochastic_UCP}
  \label{fig:qgan_qaoa_panels}
\end{figure*}

In this subsection, we evaluate the scaling of the quantum-circuit size of our method (qGAN-QAOA) with respect to the uncertainty discretization, i.e., the scenario count $N$, and the number of generators $M$.
As evaluation metrics, we use (i) the compiled gate count and (ii) the compiled depth.
Here, $N$ is related to the number of qubits in the scenario register $n_{\xi}$ by $N=2^{n_{\xi}}$, and increasing $N$ corresponds to improving the representation accuracy of the uncertainty distribution.
We denote the QAOA repetition numbers by $p_{\mathrm{1st}}$ and $p_{\mathrm{2nd}}$, and for qGAN we use a TwoLocal circuit with $p_{\mathrm{qGAN}}=n_{\xi}$.
The compiled gate count and compiled depth are evaluated using Qiskit's \texttt{transpile}.
We use \texttt{AerSimulator} as the backend, set the basis-gate set to
\(\{\mathrm{rz},\mathrm{sx},\mathrm{x},\mathrm{cx}\}\), and set the optimization level to \texttt{optimization\_level}=0.
The gate count is computed by excluding measurements and barriers from the transpiled circuit, and the circuit depth is computed using \texttt{depth()} of the transpiled circuit.

As shown in Fig.~\ref{fig:scaling_stochastic_UCP}(\subref{fig:qgan_only_gate}) and (\subref{fig:qgan_only_depth}), the circuit size of qGAN alone increases monotonically as $N$ increases.
This is because, under the setting $p_{\mathrm{qGAN}}=n_{\xi}$, the circuit size grows as a polynomial in $n_{\xi}$, i.e., it increases polylogarithmically with respect to $N$.
Therefore, with this choice of the qGAN generator circuit, the quantum-circuit size of the proposed method does not violate the $O(\mathrm{poly}\log N)$ scaling.

Next, we show the scaling when excluding the qGAN contribution from the proposed method in Fig.~\ref{fig:scaling_stochastic_UCP}(\subref{fig:no_qgan_p1_gate}) and (\subref{fig:no_qgan_p1_depth}), as well as Fig.~\ref{fig:scaling_stochastic_UCP}(\subref{fig:no_qgan_p2_gate}) and (\subref{fig:no_qgan_p2_depth}).
Figures~\ref{fig:scaling_stochastic_UCP}(\subref{fig:no_qgan_p1_gate}) and \ref{fig:scaling_stochastic_UCP}(\subref{fig:no_qgan_p1_depth}) correspond to varying $p_{\mathrm{1st}}$, where both the gate count and the depth increase approximately proportionally to $p_{\mathrm{1st}}$, and the increment is independent of the scenario count $N$.
This is because, to satisfy the non-anticipativity constraints Req.~\ref{req:Quantum_non-anticipativity}, the first-stage optimization circuit associated with $p_{\mathrm{1st}}$ contains no gates that interact with the scenario register.

In contrast, the dependence on $p_{\mathrm{2nd}}$ in Fig.~\ref{fig:scaling_stochastic_UCP}(\subref{fig:no_qgan_p2_gate}) and \ref{fig:scaling_stochastic_UCP}(\subref{fig:no_qgan_p2_depth}) is more pronounced than in the case of $p_{\mathrm{1st}}$, and the increase in the gate count and depth becomes larger as the scenario count $N$ increases.
This is because the second-stage optimization circuit introduces quantum mapping gates that interact with the scenario register.
However, even in this case, the circuit size is at most $O(\mathrm{poly}\log N)$.

Finally, for the entire proposed method including qGAN ($p_{\mathrm{1st}}=p_{\mathrm{2nd}}=1$), we show the scaling with respect to $M$ in Fig.~\ref{fig:scaling_stochastic_UCP}(\subref{fig:m_scale_gate}) and Fig.~\ref{fig:scaling_stochastic_UCP}(\subref{fig:m_scale_depth}).
Both the gate count and the depth increase approximately polynomially as $M$ increases, reflecting the additive contribution of the gate operations required for each generator.
Moreover, even for the same $M$, the circuit size increases as $N$ becomes larger, and for large $N$ in particular, the qGAN circuit component raises the baseline circuit size.
Overall, although the circuit size of our method increases with $N$ and $M$, we have $n_{\xi}=\log N$ qubits for qGAN, $n_{1\mathrm{st}}=M$ qubits for the first stage, and $n_{2\mathrm{nd}}=M$ qubits for the second stage.
Therefore, with the total number of qubits $n=n_{\xi}+n_{1\mathrm{st}}+n_{2\mathrm{nd}}$, the circuit size is at most $O(\mathrm{poly}\,n)$.
For scaling to larger problems, reducing the qGAN circuit size is important.

\subsection{Evaluation of computational complexity}
\label{subsec:computational_complexity}

In this section, we consider a setting where observational data are given for an uncertainty $\xi$ that follows an unknown population distribution, and evaluate the computational complexity---including distribution learning, expectation estimation, and two-stage stochastic programming---from the viewpoint of achieving accuracy $O(\varepsilon)$.
Here, we do not consider quantum computational errors arising from hardware noise, but focus on the query complexity required for expectation estimation and the scaling of the circuit size for a single expectation evaluation.
Moreover, since the distribution-learning error constitutes a separate axis that depends on the data size, we assume in this section that the model error $\varepsilon_{\mathrm{model}}$ is sufficiently small, and compare classical methods with the proposed method in the regime where the discretization error and the expectation-estimation error are dominant.

\subsubsection{Error decomposition}
For the first-stage objective $f(\boldsymbol{x})$ and the second-stage recourse term $Q(\boldsymbol{x},\xi)$, define the true objective value as
\[
J(\boldsymbol{x}) := f(\boldsymbol{x}) + \mathbb{E}_{\xi}[Q(\boldsymbol{x},\xi)].
\]
The difference between $J(\boldsymbol{x})$ and the estimate $\widehat J(\boldsymbol{x})$ returned by the algorithm can be upper-bounded, excluding quantum computational errors, as
\begin{equation}
\label{eq:error_budget_additive}
  |J(\boldsymbol{x})-\widehat J(\boldsymbol{x})|
  \le
  \varepsilon_{\mathrm{model}}
  + \varepsilon_{\mathrm{disc}}
  + \varepsilon_{\mathrm{est}}
  + \varepsilon_{\mathrm{opt}}
\end{equation}
Here, $\varepsilon_{\mathrm{disc}}$ denotes the discretization error due to a finite representation, $\varepsilon_{\mathrm{est}}$ denotes the statistical error of expectation estimation, and $\varepsilon_{\mathrm{opt}}$ denotes the error arising from the stopping accuracy of the variational-parameter updates by a classical optimizer.
In what follows, we evaluate the computational complexity required to achieve, in particular, $\varepsilon_{\mathrm{disc}}=O(\varepsilon)$ and $\varepsilon_{\mathrm{est}}=O(\varepsilon)$.

\subsubsection{Classical baseline: scenario Monte Carlo}
\label{subsubsec:classical_eps}

In practical classical two-stage stochastic programming, as expressed in Eq.~\eqref{eq:two_stage_sp_discrete}, scenarios $\xi_s$ are sampled from a generative model such as a GAN, and SAA approximates the expectation by the sample average.
Let the number of samples be $N_{\mathrm{MC}}$.
Based on the variance of the sample mean, we have
\begin{equation}
  \varepsilon_{\mathrm{est}} = O(N_{\mathrm{MC}}^{-1/2})
  \Rightarrow
  N_{\mathrm{MC}} = O(\varepsilon^{-2})
\end{equation}
In two-stage stochastic programming based on the L-shaped method, the second-stage subproblem is solved for each sampled scenario, and the master problem is updated using the obtained information.
Let $T_{\mathrm{2nd}}^{\mathrm{cl}}$ be the computational cost of solving one second-stage subproblem.
Then, the computational cost of evaluating the expected recourse cost per iteration of the L-shaped method scales as
\begin{equation}
\label{eq:classical_cost_per_iter}
  \mathrm{Cost}^{\mathrm{cl}}_{\mathrm{eval}}
  =
  O(N_{\mathrm{MC}}\,T_{\mathrm{2nd}}^{\mathrm{cl}})
  =
  O(\varepsilon^{-2} T_{\mathrm{2nd}}^{\mathrm{cl}})
\end{equation}
Let $K^{\mathrm{cl}}(\varepsilon)$ denote the number of iterations of the L-shaped method required to achieve $\varepsilon_{\mathrm{opt}}=O(\varepsilon)$.
Then, the total computational cost is
\begin{equation}
\label{eq:classical_total_cost}
  \mathrm{Cost}^{\mathrm{cl}}_{\mathrm{total}}
  =
  K^{\mathrm{cl}}(\varepsilon)\cdot O(\varepsilon^{-2} T_{\mathrm{2nd}}^{\mathrm{cl}})
\end{equation}

\subsubsection{qGAN-QAOA approach}
\label{subsubsec:proposed_eps}

In the proposed method, to handle the continuous uncertainty $\xi$ on a finite-dimensional scenario register, we use $N=2^{n_\xi}$ realizations obtained by uniform discretization.
Under the discretization based on Eq.~\eqref{eq:disc_rv}, we have
\begin{equation}
  \varepsilon_{\mathrm{disc}}=O(1/N)
  \Rightarrow
  N=O(\varepsilon^{-1})
\end{equation}
In this case, by Eq.~\eqref{eq:scaling_qGAN-QAOA}, the proposed quantum circuit $U_{\mathrm{qGAN\text{-}QAOA}}$ can be implemented, for $N$ scenarios, with a gate count
\begin{equation}
  C(N)=O(\mathrm{poly}(\log N))
\end{equation}
Here, $C(N)$ denotes the circuit size per shot, and factors depending on the number of terms on the decision-variable side and the QAOA depths $p_{1\mathrm{st}},p_{2\mathrm{nd}}$ are included in $C(N)$.

To perform expectation estimation by measurement sampling and achieve $\varepsilon_{\mathrm{est}}=O(\varepsilon)$, the number of shots $S=O(\varepsilon^{-2})$ is required.
Therefore, the objective-evaluation cost per variational-parameter update is
\begin{equation}\label{eq:prop_noqae_cost}
\begin{aligned}
  \mathrm{Cost}^{\mathrm{prop(no\,QAE)}}_{\mathrm{eval}}
  &=
  O(\varepsilon^{-2})\cdot C(N)\\
  &=
  O(\varepsilon^{-2})\cdot O(\mathrm{poly}(\log N))
\end{aligned}
\end{equation}
Substituting $N=O(\varepsilon^{-1})$ yields
\begin{equation}
\label{eq:prop_noqae_cost_eps}
  \mathrm{Cost}^{\mathrm{prop(no\,QAE)}}_{\mathrm{eval}}
  =
  O\bigl(\varepsilon^{-2}\mathrm{poly}(\log(\varepsilon^{-1}))\bigr)
\end{equation}

On the other hand, as also proposed by Rotello et al.~\cite{rotello2024_expected_value}, to achieve the estimation error $\varepsilon_{\mathrm{est}}=O(\varepsilon)$ using QAE (quantum amplitude estimation)~\cite{brassard2002qae}, the number of queries to the Grover-type amplification operator $Q$ repeated internally by QAE scales as $O(\varepsilon^{-1})$.
Here, $Q$ is defined using the state-preparation circuit $\mathcal{A}$ and the phase-flip gates $S_0,S_\chi$ as $Q:=\mathcal{A}S_0\mathcal{A}^\dagger S_\chi$, where $\mathcal{A}$ corresponds to the proposed circuit $U_{\mathrm{qGAN\text{-}QAOA}}$ together with auxiliary circuits for evaluation.
Since the estimation is probabilistic, to suppress the failure probability to at most $\delta$, an additional factor of $O(\log(\delta^{-1}))$ repetitions is required via independent repetitions and amplification by the median (or majority vote).
Moreover, implementing $Q$ requires the controlled circuit $\mathrm{c}\text{-}U_{\mathrm{qGAN\text{-}QAOA}}$ and the adjoint circuit $U_{\mathrm{qGAN\text{-}QAOA}}^{\dagger}$ of $U_{\mathrm{qGAN\text{-}QAOA}}$; we aggregate the resulting increase in gate count as well as the $O(\log(\delta^{-1}))$ repetitions into an overhead factor $C_{\mathrm{AE}}$.
Therefore, the objective-evaluation cost per variational-parameter update is
\begin{equation}
\label{eq:prop_qae_cost}
\begin{aligned}
  \mathrm{Cost}^{\mathrm{prop(QAE)}}_{\mathrm{eval}}
  &=
  O(\varepsilon^{-1})\cdot C_{\mathrm{AE}}\cdot C(N)\\
  &=
  O(\varepsilon^{-1})\cdot C_{\mathrm{AE}}\cdot O(\mathrm{poly}(\log N))
\end{aligned}
\end{equation}
Substituting $N=O(\varepsilon^{-1})$ yields
\begin{equation}
\label{eq:prop_qae_cost_eps}
  \mathrm{Cost}^{\mathrm{prop(QAE)}}_{\mathrm{eval}}
  =
  O\bigl(\varepsilon^{-1} \cdot C_{\mathrm{AE}} \cdot \mathrm{poly}(\log \varepsilon^{-1})\bigr)
\end{equation}
Let $K^{\mathrm{prop}}(\varepsilon)$ denote the number of variational-parameter updates required to achieve $\varepsilon_{\mathrm{opt}}=O(\varepsilon)$.
Then, the total computational cost is given by multiplying each evaluation cost by $K^{\mathrm{prop}}(\varepsilon)$.

\subsubsection{Comparison of classical methods and qGAN-QAOA}
\label{subsubsec:eps_asymptotic_compare}

\begin{table}[t]
\centering
\small
\renewcommand{\arraystretch}{1.35}
\caption{Computational complexity per iteration for estimating the objective value with accuracy $O(\varepsilon)$.}
\label{tab:complexity_comparison_eq}
\begin{tabular}{l @{\hspace{0.5cm}} l}
\hline
\hline
Method & Complexity per iteration \\
\hline
Classical SP (SAA) &
$O(\varepsilon^{-2}\cdot T_{\mathrm{2nd}}^{\mathrm{cl}})$ \\
qGAN-QAOA (shot-based) &
$O\bigl(\varepsilon^{-2}\cdot \mathrm{poly}(\log \varepsilon^{-1})\bigr)$ \\
qGAN-QAOA (with QAE) &
$O\bigl(\varepsilon^{-1} \cdot C_{\mathrm{AE}} \cdot \mathrm{poly}(\log \varepsilon^{-1}\bigr)$ \\
\hline
\hline
\end{tabular}
\end{table}

Here, we do not take into account the difference between the iteration counts $K^{\mathrm{cl}}(\varepsilon)$ and $K^{\mathrm{prop}}(\varepsilon)$ required under $\varepsilon_{\mathrm{est}}=O(\varepsilon)$, and instead focus the comparison on the differences in the $\varepsilon$-dependence and the scenario-count dependence of the objective-evaluation cost per iteration.
In this setting, the leading $\varepsilon$-dependence of the objective-evaluation cost per iteration is summarized in Table~\ref{tab:complexity_comparison_eq}.

The classical method requires $O(\varepsilon^{-2})$ samples for expectation estimation.
Since the proposed method can compress the scenario dependence of the quantum circuit to $\mathrm{poly}(\log N)$, even under the discretization accuracy $N=O(\varepsilon^{-1})$, the additional dependence on $N$ is suppressed to a logarithmic factor.
Moreover, when combined with QAE, the query complexity for expectation estimation improves from $O(\varepsilon^{-2})$ to $O(\varepsilon^{-1})$, suggesting that the proposed method can become advantageous when $\varepsilon$ is very small.

Overall, the advantages offered by the proposed method can be summarized as follows:
it can treat the expected recourse as a circuit expectation value while retaining uncertainty as a distribution, and it can compress the $N$-dependence on the scenario-register side to a logarithmic scale by exploiting the structure of uniform discretization.
On the other hand, whether this advantage materializes depends on several assumptions, including:
(1) the problem structure is such that the expectation-evaluation (scenario-enumeration) component induced by increasing the scenario count $N$ is dominant in the total computational cost;
(2) the product of the required number of shots for expectation estimation and the number of objective-function evaluations (iterations) needed for variational optimization remains within available computational resources; and
(3) the qGAN training error is sufficiently small so that it does not become dominant relative to other error sources such as discretization and measurement errors.
Therefore, for practical scaling to larger problems, the key is to design a scenario-generator circuit that achieves the required distribution-representation accuracy without increasing circuit depth, and to stabilize expectation estimation and variational optimization under noise.
In the next section, we summarize our conclusions and discuss future directions.

\section{Conclusion}\label{sec:conclusion} 

In this study, we propose qGAN-QAOA, a unified quantum-circuit workflow: a pre-trained quantum generative adversarial network (qGAN) encodes the scenario distribution, and QAOA optimizes first-stage decisions by minimizing the two-stage objective, including the expected recourse cost, on a single variational quantum circuit.
The proposed method encodes a discretized scenario distribution into the amplitudes of the scenario register using a pre-trained qGAN, and, with the generator parameters $\boldsymbol{\theta}^\ast$ fixed, performs circuit-expectation minimization of the two-stage objective by optimizing the variational parameters of QAOA.
Moreover, based on the circuit structure, we showed that the marginal distribution of the first-stage measurement outcomes is independent of the scenario, and clarified that the property corresponding to the non-anticipativity constraints in two-stage stochastic programming holds in terms of measurement outcome statistics.

Furthermore, under a setting where a continuous uncertainty is uniformly discretized, we showed that the gate count and circuit depth of the proposed quantum circuit can be bounded by $O(\mathrm{poly}(\log N))$ with respect to the scenario count $N$ by exploiting that the random-variable operator $\hat{\xi}$ admits a sparse Pauli-$Z$ expansion via the Walsh--Hadamard transform.
This result theoretically positions the proposed method as having room for improvement in its dependence on the scenario count, particularly in regimes where scenario enumeration in SAA becomes a bottleneck.
As a case study, we considered a two-stage UCP with PV output uncertainty and confirmed that qGAN can learn the distribution with high agreement, that qGAN-QAOA concentrates the first-stage solutions into a small set of promising candidates, and that the expected cost of a representative solution tends to be closer to the stochastic-programming baseline (RP) than to the expected-value baseline (EEV).
Overall, the proposed method is promising as a quantum-circuit implementation of two-stage decision-making that treats uncertainty as a distribution.

For future work, improving the efficiency of the scenario-generator circuit and enhancing the distribution-approximation accuracy are important.
Beyond TwoLocal, designing ans\"atze that reflect structural properties of the target distribution---such as skewness, tail heaviness, and multimodality---may further reduce the required circuit depth.
Next, validation on real quantum hardware is essential.
Since this paper is primarily based on simulator experiments, it is necessary to quantify performance degradation under noise, in particular the impact of error accumulation in multi-qubit, multi-layer QAOA, and to clarify practical prospects in combination with circuit optimization and error mitigation.
Moreover, extending the approach to stochastic-programming problems beyond stochastic UCP is also important.
For problems that can be formulated as two-stage stochastic programming, such as supply-chain design and portfolio optimization, we plan to investigate the effectiveness and limitations of this integrated framework: a pre-trained qGAN encodes the scenario distribution, and QAOA optimizes first-stage decisions by minimizing the two-stage objective, including the expected recourse cost.
In addition, to address more general decision-making processes, extending the proposed framework to multi-stage stochastic programming remains an important direction for future research.

\section*{Acknowledgements}
This work was partially supported by the New Energy and Industrial Technology Development Organization (NEDO) under grant number JPNP23003. 
We are grateful to Anthony Wilkie for insightful discussions and feedback during the INFORMS Annual Meeting 2025.


\bibliographystyle{apsrev4-2}  
\bibliography{ref.bib} 

\clearpage
\appendix
\section{Proof of Proposition~\ref{prop:quantum_two_stage_objective}}
\label{app:proof_quantum_two_stage_objective}

Since the total Hamiltonian is given by
\begin{equation}
  H_P
  =
  H_P^{1\mathrm{st}}(\hat{\boldsymbol{x}}^{1\mathrm{st}})
  +
  H_P^{2\mathrm{nd}}(\hat{\boldsymbol{y}}^{2\mathrm{nd}},\hat{\boldsymbol{x}}^{1\mathrm{st}},\hat{\xi})
  \label{eq:app_HP_split}
\end{equation}
the expectation value $\braket{\Psi | H_P | \Psi}$ can be decomposed as
\begin{equation}
  \braket{\Psi | H_P | \Psi}
  =
  \braket{\Psi | H_P^{1\mathrm{st}} | \Psi}
  +
  \braket{\Psi | H_P^{2\mathrm{nd}} | \Psi}.
  \label{eq:app_HP_expect_split}
\end{equation}
Because $H_P^{1\mathrm{st}}$ acts only on the first-stage register and
$  \hat{\boldsymbol{x}}^{1\mathrm{st}}
  \ket{b_k^{1\mathrm{st}}}
  = \boldsymbol{x}_k^{1\mathrm{st}}\ket{b_k^{1\mathrm{st}}}$ holds, $\ket{b_k^{1\mathrm{st}}}$ is an eigenstate and thus
\begin{equation}
  H_P^{1\mathrm{st}}(\hat{\boldsymbol{x}}^{1\mathrm{st}})
  \ket{b_k^{1\mathrm{st}}}
  =
  H_P^{1\mathrm{st}}(\boldsymbol{x}_k^{1\mathrm{st}})
  \ket{b_k^{1\mathrm{st}}}
  \label{eq:app_HP_1st}
\end{equation}
follows.
Using the qGAN-QAOA ansatz $|\Psi\rangle$ in Eq.~\eqref{eq:psi_final_expanded}, we obtain
\begin{equation}
\begin{aligned}
  \braket{\Psi | H_P^{1\mathrm{st}} | \Psi}
  &=
  \sum_{k,k'}
  \sum_{s,s'}
    \alpha_k \alpha_{k'}^\ast
    \sqrt{p_s p_{s'}}
  \nonumber \\
  &
    \bigl(
      \bra{\psi_{s',k'}^{2\mathrm{nd}}}
      \otimes \bra{b_{k'}^{1\mathrm{st}}}
      \otimes \bra{b_{s'}^\xi}
    \bigr)
    H_P^{1\mathrm{st}}\\
  &\hspace{5em}\times
    \bigl(
      \ket{\psi_{s,k}^{2\mathrm{nd}}}
      \otimes \ket{b_k^{1\mathrm{st}}}
      \otimes \ket{b_s^\xi}
    \bigr).
\end{aligned}
\end{equation}
Since $H_P^{1\mathrm{st}}$ acts only on the first-stage register and
\(
  \braket{b_{k'}^{1\mathrm{st}} | b_k^{1\mathrm{st}}}
  = \delta_{k'k},
  \;
  \braket{b_{s'}^\xi | b_s^\xi}
  = \delta_{s's}
\),
we have
\begin{equation}
\begin{aligned}
  \braket{\Psi | H_P^{1\mathrm{st}} | \Psi}
  &=
  \sum_{k}
  \sum_{s}
    |\alpha_k|^2 p_s\,
    \braket{\psi_{s,k}^{2\mathrm{nd}} | \psi_{s,k}^{2\mathrm{nd}}}
    H_P^{1\mathrm{st}}(\boldsymbol{x}_k^{1\mathrm{st}}) \\
  &=
  \sum_{k}
    |\alpha_k|^2
    H_P^{1\mathrm{st}}(\boldsymbol{x}_k^{1\mathrm{st}})
\end{aligned}
\label{eq:app_HP1_expect}
\end{equation}
where we used
$\langle\psi_{s,k}^{2\mathrm{nd}}|\psi_{s,k}^{2\mathrm{nd}}\rangle = 1
  ,\,\sum_{s=0}^{N-1} p_s = 1$
and Eq.~\eqref{eq:app_HP_1st}.
Similarly, for the second-stage term we can write
\begin{equation}
\begin{aligned}
  \braket{\Psi | H_P^{2\mathrm{nd}} | \Psi}
  &=
  \sum_{k,k'}
  \sum_{s,s'}
    \alpha_k \alpha_{k'}^\ast
    \sqrt{p_s p_{s'}}
  \\
  &\times
    \bigl(
      \bra{\psi_{s',k'}^{2\mathrm{nd}}}
      \otimes \bra{b_{k'}^{1\mathrm{st}}}
      \otimes \bra{b_{s'}^\xi}
    \bigr) 
  \\
  &\times
    H_P^{2\mathrm{nd}}(\hat{\boldsymbol{y}}^{2\mathrm{nd}},\hat{\boldsymbol{x}}^{1\mathrm{st}},\hat{\xi})
    \bigl(
      \ket{\psi_{s,k}^{2\mathrm{nd}}}
      \otimes \ket{b_k^{1\mathrm{st}}}
      \otimes \ket{b_s^\xi}
    \bigr)
\end{aligned}
\end{equation}
Although $H_P^{2\mathrm{nd}}$ acts on all three registers, the orthogonality of the basis states again leaves only the terms with $k'=k$ and $s'=s$, yielding
\begin{equation}
\begin{aligned}
  \braket{\Psi | H_P^{2\mathrm{nd}} | \Psi}
  &=
  \sum_{k}
  \sum_{s}
    |\alpha_k|^2 p_s\,
  \nonumber \\
  &
    \bra{\psi_{s,k}^{2\mathrm{nd}}}
      H_P^{2\mathrm{nd}}(\hat{\boldsymbol{y}}^{2\mathrm{nd}},\boldsymbol{x}_k^{1\mathrm{st}},\xi_s)
    \ket{\psi_{s,k}^{2\mathrm{nd}}}.
\end{aligned}
\label{eq:app_HP2_expect_raw}
\end{equation}
Defining
\begin{equation}
  Q(\boldsymbol{x}_k^{1\mathrm{st}},\xi_s)
  :=
  \bra{\psi_{s,k}^{2\mathrm{nd}}}
    H_P^{2\mathrm{nd}}(\hat{\boldsymbol{y}}^{2\mathrm{nd}},\boldsymbol{x}_k^{1\mathrm{st}},\xi_s)
  \ket{\psi_{s,k}^{2\mathrm{nd}}}
  \label{eq:app_Q_def}
\end{equation}
we can rewrite Eq.~\eqref{eq:app_HP2_expect_raw} as
\begin{equation}
\begin{aligned}
  \braket{\Psi | H_P^{2\mathrm{nd}} | \Psi}
  &=
  \sum_{k}
    |\alpha_k|^2
    \sum_{s}
      p_s\,
      Q(\boldsymbol{x}_k^{1\mathrm{st}},\xi_s) \\
  &=
  \sum_{k}
    |\alpha_k|^2\,
    \mathbb{E}_\xi\bigl[Q(\boldsymbol{x}_k^{1\mathrm{st}},\xi)\bigr]
\end{aligned}
\label{eq:app_HP2_expect}
\end{equation}
Substituting Eqs.~\eqref{eq:app_HP1_expect} and~\eqref{eq:app_HP2_expect} into Eq.~\eqref{eq:app_HP_expect_split}, we obtain
\begin{equation}
\begin{aligned}
  \braket{\Psi | H_P | \Psi}
  &=
  \sum_{k}
    |\alpha_k|^2
    H_P^{1\mathrm{st}}(\boldsymbol{x}_k^{1\mathrm{st}})
  +
  \sum_{k}
    |\alpha_k|^2\,
    \mathbb{E}_\xi\bigl[Q(\boldsymbol{x}_k^{1\mathrm{st}},\xi)\bigr] \\
  &=
  \sum_{k}
    |\alpha_k|^2
    \left[
      H_P^{1\mathrm{st}}(\boldsymbol{x}_k^{1\mathrm{st}})
      +
      \mathbb{E}_\xi\bigl[Q(\boldsymbol{x}_k^{1\mathrm{st}},\xi)\bigr]
    \right]
\end{aligned}
\end{equation}
which yields Eq.~\eqref{eq:quantum_two_stage_obj}.

\section{Proof of Proposition~\ref{prop:arithmetic_WHT}}
\label{app:proof_arithmetic_WHT}
\noindent\textbf{Proof.}
Here, we provide a detailed proof of Proposition~\ref{prop:arithmetic_WHT}.
From Eq.~\eqref{eq:arithmetic_seq}, the scenario realizations can be written as
  \begin{equation*}
    \xi_s
    = \xi_{\min} + \Delta\xi\, s,
    \quad s=0,\dots,N-1
  \end{equation*}
By the definition of the Walsh--Hadamard transform in Eq.~\eqref{eq:WHT_coeff_general}, the coefficient $c_j$ is given by
  \begin{equation*}
    c_j
    = \frac{1}{N}
      \sum_{s=0}^{N-1}
        (-1)^{j\cdot s}\,\xi_s,
    \quad j=0,\dots,N-1
  \end{equation*}
Using the linearity of $\xi_s$, we obtain
  \begin{equation*}
  \begin{aligned}
    c_j
    &= \frac{\xi_{\min}}{N}
        \sum_{s=0}^{N-1} (-1)^{j\cdot s}
      + \frac{\Delta\xi}{N}
        \sum_{s=0}^{N-1} s\,(-1)^{j\cdot s}\\
    &=: \frac{\xi_{\min}}{N} S^{(0)}_j
      + \frac{\Delta\xi}{N} S^{(1)}_j      
  \end{aligned}
  \end{equation*}
where we define
  \begin{align*}
    S^{(0)}_j &:= \sum_{s=0}^{N-1} (-1)^{j\cdot s},\\
    S^{(1)}_j &:= \sum_{s=0}^{N-1} s\,(-1)^{j\cdot s}.
  \end{align*}

First, for $S^{(0)}_j$, the orthogonality of the row vectors of the Walsh--Hadamard matrix $H_{n_\xi}$ implies
  \begin{equation*}
    S^{(0)}_j
    = \sum_{s=0}^{N-1} (-1)^{j\cdot s}
    =
    \begin{cases}
      N, & j=0,\\
      0, & j\neq 0,
    \end{cases}
  \end{equation*}
and therefore the term arising from $\xi_{\min}$ contributes only when $j=0$, yielding
  \begin{equation*}
  \begin{aligned}
      \text{for } j=0:\quad
    c_0^{(\mathrm{const})}
    &= \frac{\xi_{\min}}{N} S^{(0)}_0
    = \xi_{\min},\\
    \text{for } j\neq 0:\quad
    c_j^{(\mathrm{const})}
    &= 0.
  \end{aligned}
  \end{equation*}

Next, we evaluate $S^{(1)}_j$.
Writing the binary representation of $s$ as
  \begin{equation*}
    s = \sum_{i=0}^{n_\xi-1} 2^i b_{s,i},
    \qquad b_{s,i}\in\{0,1\},
  \end{equation*}
we have
  \begin{equation*}
  \begin{aligned}
    S^{(1)}_j
    &= \sum_{s=0}^{N-1}
        \left(
          \sum_{i=0}^{n_\xi-1} 2^i b_{s,i}
        \right)
        (-1)^{j\cdot s}\\
    &=
    \sum_{i=0}^{n_\xi-1} 2^i
      \sum_{s=0}^{N-1}
        b_{s,i} (-1)^{j\cdot s}\\
    &=: \sum_{i=0}^{n_\xi-1} 2^i S^{(1)}_{j,i}.      
  \end{aligned}
  \end{equation*}
To evaluate the inner sum
  \(
    S^{(1)}_{j,i}
    := \sum_{s=0}^{N-1}
         b_{s,i} (-1)^{j\cdot s},
  \)
we count $s$ by separating the $i$-th bit from the remaining bits.
That is, we regard $s$ as
  \(
    s = (t,b_{s,i}),
  \)
where $t$ ranges over all assignments of the other $n_\xi-1$ bits.
In this case, there exists a function $\phi_j(t)$ such that
  \[
    j\cdot s
    = j_i b_{s,i}
      + \sum_{k\neq i} j_k b_{s,k}
    = j_i b_{s,i} + \phi_j(t).
  \]
Therefore,
  \begin{align*}
    S^{(1)}_{j,i}
    &= \sum_{t}
       \Bigl[
         \text{the term with } b_{s,i}=0
         + \text{the term with } b_{s,i}=1
       \Bigr] \\
    &= \sum_{t}
       \Bigl[
         0 \cdot (-1)^{\phi_j(t)}
         + 1 \cdot (-1)^{j_i + \phi_j(t)}
       \Bigr] \\
    &= (-1)^{j_i}
       \sum_{t} (-1)^{\phi_j(t)}.
  \end{align*}

If $j_i=0$, then
  \(
    S^{(1)}_{j,i}
    = \sum_t (-1)^{\phi_j(t)}.
  \)
Since $t$ ranges over all $n_\xi-1$ bits, this is analogous to a Walsh--Hadamard sum over $(n_\xi-1)$ bits and becomes zero unless all bits of $j$ except the $i$-th are zero (i.e., $j=0$).
On the other hand, when $j_i=1$ and all other bits $j_k$ ($k\neq i$) are zero, i.e., $j=2^i$, we have
  \(
    \phi_j(t)\equiv 0,
  \)
and thus
  \[
    S^{(1)}_{j,i}
    = (-1)^{1} \sum_t 1
    = -\,\frac{N}{2}.
  \]
Summarizing these cases, we obtain
  \[
    S^{(1)}_{j,i}
    =
    \begin{cases}
      -\,\dfrac{N}{2}, & j = 2^i,\\[0.5ex]
      0, & j\neq 2^i,
    \end{cases}
  \]
and hence
  \begin{equation*}
  \begin{aligned}
      S^{(1)}_j
    &= \sum_{i=0}^{n_\xi-1} 2^i S^{(1)}_{j,i}\\
    &=
    \begin{cases}
      \displaystyle
      \sum_{s=0}^{N-1} s
      = \dfrac{N(N-1)}{2} & j=0\\
      \displaystyle
      -\,2^{i}\,\dfrac{N}{2}
      = -\,N\,2^{i-1} & j=2^i\ (i=0,\dots,n_\xi-1)\\
      0 & w(j)\ge 2
    \end{cases}
  \end{aligned}
  \end{equation*}
Finally, substituting the above evaluations of $S^{(0)}_j$ and $S^{(1)}_j$ into
  \begin{equation*}
      c_j
    = \frac{\xi_{\min}}{N} S^{(0)}_j
      + \frac{\Delta\xi}{N} S^{(1)}_j,
  \end{equation*}
and noting that $N=2^{n_\xi}$, we obtain:

    \noindent\textbullet\ for $j=0$,
    \begin{align*}
      c_0
      &= \frac{\xi_{\min}}{N} N
         + \frac{\Delta\xi}{N}\,\frac{N(N-1)}{2} \\[0.3ex]
      &= \xi_{\min} + \Delta\xi\,\frac{N-1}{2} \\[0.3ex]
      &= \xi_{\min}
         + \Delta\xi\,\frac{2^{n_\xi}-1}{2}.
    \end{align*}
    
    \noindent\textbullet\ for $j=2^i$ ($i=0,\dots,n_\xi-1$),
    \begin{align*}
      c_j
      &= \frac{\xi_{\min}}{N}\cdot 0
         + \frac{\Delta\xi}{N}
           \bigl(-\,N\,2^{i-1}\bigr) \\[0.3ex]
      &= -\,\Delta\xi\,2^{i-1}.
    \end{align*}
    
    \noindent\textbullet\ for $w(j)\ge 2$,
    \begin{align*}
      c_j &= 0 + 0 = 0.
    \end{align*}

These coincide with Eq.~\eqref{eq:WHT_coeff_arith} in Proposition~\ref{prop:arithmetic_WHT}, completing the proof.

\end{document}